\documentclass{velds}

\usepackage{comment}
\usepackage{subcaption}
\usepackage{balance}
\usepackage{graphicx}
\usepackage{amsmath}
\usepackage{amssymb}
\usepackage{algorithm}
\usepackage[noend]{algpseudocode}
\usepackage{multirow}
\usepackage{url}
\usepackage{xcolor}
\usepackage{xspace}
\usepackage{hyphenat}
\usepackage[font=small,labelfont=bf]{caption}
\usepackage{tikz,pgfplots,pgfplotstable}
\usepackage[normalem]{ulem}
\usepackage{upgreek}
\usepackage{hyperref}

\usetikzlibrary{arrows}
\usetikzlibrary{positioning}
\usetikzlibrary{patterns}
\usetikzlibrary{shapes}

\pgfplotsset{compat=newest}

\renewcommand{\algorithmicrequire}{\textbf{Input: }}
\renewcommand{\algorithmicensure}{\textbf{Output: }}
\DeclareMathOperator*{\argmax}{argmax}
\DeclareMathOperator*{\argmin}{argmin}

\newtheorem{thm}{Theorem}

\newtheorem{defn}{Definition}
\newtheorem{lem}{Lemma}
\newtheorem{cor}{Corollary}

\newtheorem{prob}{Problem}
\newtheorem{example}{Example}
\newtheorem{prop}{Proposition}

\newcommand{\Var}{\mathrm{Var}}
\newcommand{\Cov}{\mathrm{Cov}}
\newcommand{\INF}{\hat{INF}_F}
\newcommand{\model}{TCIC\xspace}
\newcommand{\stitle}[1]{\vspace{3pt}\noindent{\bf #1}}
\newcommand{\etitle}[1]{\vspace{1ex}\noindent{\emph{#1}}}

\newcommand{\cyan}[1]{\textcolor{cyan}{#1}}
\newcommand{\red}[1]{\textcolor{red}{#1}}
\newcommand{\edit}[1]{{#1}}
\newcommand{\LL}[1]{{#1}}
\newcommand{\eat}[1]{} 
\newcommand{\laks}[1]{#1}

\vldbTitle{Misinformation Mitigation under Differential Propagation Rates and Temporal Penalties}
\vldbAuthors{Michael Simpson, Laks V.S. Lakshmanan, Farnoosh Hashemi}
\vldbDOI{https://doi.org/xx.xxxxx/xxxxxxx.xxxxxxx}
\vldbVolume{12}
\vldbNumber{xxx}
\vldbYear{2019}

\title{Misinformation Mitigation under Differential Propagation Rates and Temporal Penalties}

\numberofauthors{1}
\author
{
    \alignauthor
    Michael Simpson\qquad
    Farnoosh Hashemi\qquad
    Laks V.S. Lakshmanan\\
    \affaddr
    {
        University of British Columbia
    }
    {
        \email
        {
            \{mesimp,farsh,laks\}@cs.ubc.ca
        }
    }
}

\begin{document}

\maketitle

\begin{abstract}
We propose an information propagation model that captures important temporal aspects that have been well observed in the dynamics of fake news diffusion, in contrast with the diffusion of truth. The model accounts for differential propagation rates of truth and misinformation and for user reaction times. We study a time-sensitive variant of the \textit{misinformation mitigation} problem, where $k$ seeds are to be  selected to activate a truth campaign so as to minimize the number of users that adopt misinformation propagating through a social network. We show that the resulting objective is non-submodular and employ a sandwiching technique by defining submodular upper and lower bounding functions, providing data-dependent guarantees. In order to enable the use of a reverse sampling framework, we introduce a weighted version of reverse reachability sets that captures the associated differential propagation rates and establish a key equivalence between weighted set coverage probabilities and mitigation with respect to the sandwiching functions. Further, we propose an offline reverse sampling framework that provides $(1 - 1/e - \epsilon)$-approximate solutions to our bounding functions and introduce an importance sampling technique to reduce the sample complexity of our solution. Finally, we show how our framework can provide an anytime solution to the problem. Experiments over five datasets show that our approach outperforms previous approaches and is robust to uncertainty in the model parameters.
\end{abstract}

\section{Introduction}
\label{sec:intro}
\begin{figure}
\centering
\begin{tikzpicture}

\tikzset{node/.style={circle,draw,minimum size=0.5cm,inner sep=0pt},}
\tikzset{edge/.style={->,> = latex'},}

\node[node] at (0,0) (1) {$v_0$};
\node[node] at (1.5,-0.4) (2) {$v_1$};
\node[node] at (1.25,0.75) (3) {$v_2$};
\node[node, label={[red]above:$4$}] at (-0.25,1) (4) {$v_3$};
\node[node, label={[red]left: $1$}] at (-1,0.5) (5) {$v_4$};
\node[node] at (-0.75,-0.75) (6) {$v_5$};
\node[node] at (0.6,-1) (7) {$v_6$};
\node[node] at (2.5,0.75) (8) {$v_7$};
\node[node] at (1.75,1.75) (9) {$v_8$};
\node[node] at (0.25,1.75) (10) {$v_9$};
\node[node] at (-1.75,1.25) (11) {$v_{10}$};
\node[node] at (-1.75,-0.2) (12) {$v_{11}$};
\node[node] at (-2.5,0.2) (13) {$v_{12}$};
\node[node] at (2.5,-0.9) (14) {$v_{13}$};
\node[node] at (-3.25,1) (15) {$v_{14}$};
\node[node] at (-2,-1) (16) {$v_{15}$};

\path[draw,thick,->,> = latex']
(1) edge node {} (2)
(1) edge node {} (3)
(1) edge node {} (4)
(1) edge node {} (5)
(1) edge node {} (6)
(1) edge node {} (7)
(2) edge node[left, cyan] {5} (3)
(2) edge node{} (8)
(3) edge node[left, cyan] {3} (9)
(3) edge node {} (10)
(4) edge node {} (10)
(11) edge node {} (4)
(5) edge node {} (12)
(12) edge node {} (5) 
(6) edge node {} (16)
(8) edge node[right, cyan] {3} (9)
(6) edge node {} (7)
(11) edge node {} (5)
(13) edge node {} (12)
(13) edge node[left, cyan] {2} (16)
(14) edge node{} (2)
(15) edge node[above, cyan] {3} (11)
(13) edge node[below left, cyan] {2} (15);

\path[draw,thick,<->,> = latex']
(6) edge node {} (12)
(4) edge node {} (3)
(7) edge node {} (2);

\end{tikzpicture}
\caption{Sample instance:  \cyan{edge labels} = meeting lengths (ML); MLs not shown are \cyan{$1$};  \red{numbers} besides nodes are reaction times = activation window (AW)  lengths; AW lengths not shown are \red{$0$}.}\label{fig:intro}
\end{figure}
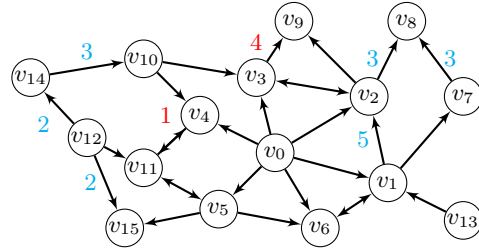
Social networks have rapidly transformed into a prominent hub for political campaigns, viral marketing, and the dissemination of news and health information. As an unfortunate side effect, there has been an increase in the number of ``bad actors'', such as spammers, hackers, and bots, exploiting these platforms to spread fake news and misinformation. A fundamental question  is: \emph{How can one limit the spread of misinformation in social networks?} Once misinformation is detected, one feasible approach is to introduce a \emph{truth campaign} with a goal of reaching users before, or at least not much later than, they are reached by the misinformation. The \emph{misinformation mitigation} (MM) problem \cite{budak2011limiting, he2012influence} aims to select effective seed nodes for the truth campaign such that the spread of misinformation can be limited as much as possible. 

Notably, the existing family of propagation models considered in the MM literature \cite{Fan2013, pham2018targeted, pham2019minimum, song2017temporal, tong2020stratlearner, tong2018misinformation, tong2019beyond, simpson2020reverse, tong2017efficient, budak2011limiting, Nguyen2012containment} fail to incorporate critical temporal aspects that have been well observed in the dynamics of fake news diffusion \cite{gabielkov2016social, mitchell2016long, vosoughi2018spread}. We argue that it is important to distinguish between the relative propagation dynamics of fake news and truth because it has been observed that fake news often ``spreads like wildfire'' online \cite{vosoughi2018spread} while the adoption of truth occurs much slower. \edit{For instance, Figure~\ref{fig:intro}  depicts a sample instance of misinformation/truth propagation, where all the edge propagation probabilities are $1$ and the edge labels indicate the diffusion time delay for truth. E.g., if node $v_2$ adopts misinformation at time $t$, it will propagate to $v_8$ at time $t+1$. By contrast, if $v_2$ adopted truth at $t$, it will propagate to $v_8$ at $t+\cyan{3}$.} The differential propagation rates of truth and fake news may have considerable consequences for the selection of effective seed nodes. Secondly, important dynamics are at play in user decision making. E.g., a recent study \cite{gabielkov2016social} observed that approximately $59\%$ of users forego reading articles linked in social media posts before acting on them, while  the time spent reading linked articles  \cite{mitchell2016long} ranges from under a minute for short articles ($<$ 250 words) to several minutes for longer articles ($\ge$ 1000 words). Thus, the ``reaction times" of users can vary considerably. \edit{To illustrate, in Figure~\ref{fig:intro}, all users save $v_3, v_4$ may react instantly (no reading), while $v_3$ and $v_4$ may read linked articles before reacting, with $v_4$ reacting quicker (\red{$1$} time unit) than $v_3$ (\red{$4$} time units).} Finally, mitigation strategies for combating misinformation are only effective when the \emph{truth} arrives at a user either (a) before the arrival of the fake news or (b) shortly after the user  becoming aware of the fake news. In particular, truth arriving \emph{later} than misinformation may still have an effect, up to a reasonable delay. \edit{E.g., suppose $v_0$ is a seed for misinformation while $v_{12}$ is a seed for truth, with both adopting the information at time $0$. Then $v_3, v_4, v_5$ will become aware of the misinformation at time $1$, while the truth reaches $v_3, v_4, v_5$ at times $6, 2, 2$ respectively. Meanwhile, $v_3, v_4, v_5$ will react at times $1+4 = 5$, $1+1 = 2$, and $1+0 = 1$ respectively. Thus, truth reaches $v_4$ in time for its reaction while it arrives at $v_3$ and $v_5$ too late.} Further, as illustrated above, as the delay between the arrival of the misinformation and the truth increases, the \emph{effectiveness} of the mitigating campaign drops off significantly. \textit{Models used in prior misinformation mitigation studies fail to account for these  phenomena, supported by real-world observations.}


\stitle{\edit{Novelty.}} \edit{In this paper, we study an interesting and realistic variant of the classic MM problem incorporating the twin time-critical aspects of misinformation propagation that have been observed and validated in earlier studies \cite{gabielkov2016social, mitchell2016long, vosoughi2018spread}: differential propagation rates and user reaction times. Our approach uses a node-level dynamic penalty function based on the delay between the arrival of the competing campaigns. We propose a new propagation model, the \textit{Temporal Competitive Independent Cascade}  (TCIC) model  which, unlike existing propagation models, accounts for differential propagation rates and user reaction times, by employing  two critical components which work jointly to properly model the dynamics of diffusion: edge-level campaign-specific \textit{time-delayed} propagation and node-level \emph{activation windows} for making adoption decisions.
} 

We then define a novel optimization problem for misinformation mitigation under the \model model. \edit{Unlike prior MM models \cite{budak2011limiting, tong2018misinformation, tong2019beyond, simpson2020reverse, tong2017efficient, song2017temporal, saxena2020mitigating}, which are based on competitive IC \cite{Chen2013}, that satisfy submodularity when the campaigns share propagation probabilities, we show that submodularity does not hold in the TCIC model even when the campaigns share probabilities. We
also show that the recent guarantees shown for Greedy \cite{bian2017guarantees} based on curvature and submodularity ratio when applied to non-submodular objectives, lead to degenerate results for our objective function.} To overcome this challenge, we employ the \emph{Sandwich Approximation (SA)} \cite{lu2015competition} and develop non-trivial upper and lower bounding submodular functions to produce solutions with a data-dependent approximation guarantee. 
Further, existing state-of-the-art solutions for \emph{influence maximization} (under a single campaign), such as \emph{IMM} \cite{tang2015influence}, \emph{SSA} \cite{nguyen2016stop, huang2017revisiting}, and \emph{OPIM} \cite{tang2018online}, are based on \emph{reverse sampling (RS)}. However, adapting the \emph{RS} machinery (based on \emph{reverse reachability (RR)} sets) to our propagation model \edit{comes with the challenge that the set of nodes reached by the misinformation is an \emph{unknown}. Due to the nature of the MM problem, our algorithm cannot inherit the expression for the sample complexity from prior work and thus requires a novel derivation, which we provide.} We develop a \emph{RS} framework by building on the state-of-the-art \emph{OPIM} algorithm. Due to the complex interactions that occur during the propagation of the fake and mitigating campaigns as well as the new temporal model components, the construction of the analog to RR sets under our new model requires great care. Further, adapting RR sets to our setting requires pushing the idea of tie-breaking between campaigns into the notion of RR sets. To the best of our knowledge, we are the \textit{first to incorporate a proportional tie-breaking rule} into the reverse sampling framework, which introduces additional challenges in constructing RR sets. Second, in order to further reduce the number of samples required by the \emph{RS} framework, we define an unbiased estimator for our upper and lower bound objectives based on \emph{importance sampling}  leading to reduced variance and tighter concentration bounds. \edit{We tackle the challenge of the sample complexity depending on the expected influence of the misinformation, an  \textit{unknown}, by developing a novel normalization term.}  

Our main contributions are as follows. (1) We introduce a novel propagation model capturing important temporal aspects pertaining to the diffusion of and reaction to truth and misinformation and define a novel MM problem formulation with delay-dependent reward (\S~\ref{sec:prelim}). (2) We develop non-trivial upper and lower bounding submodular functions for our non-submodular objective for use in a sandwiching technique (\S~\ref{sec:sandwich}). (3) We introduce an importance sampling technique to reduce the sample complexity of our solution (\S~\ref{sec:importance_sampling}). (4) We develop a reverse sampling framework that provides a $(1-1/e-\epsilon)$-approximate solution to the upper and lower bounding objectives, which in turn yields an instance-dependent approximation to the MM objective (\S~\ref{sec:reverse_sampling}). (5) We show that our algorithm can provide an anytime solution to the MM problem, with a certified approximation guarantee with bounded failure probability (\S~\ref{sec:reverse_sampling}). (7) We present a thorough experimental validation (\S~\ref{sec:experiments}). We conclude the paper in \S~\ref{sec:conclusion}.

\vspace*{-2ex} 
\section{Preliminaries}
\label{sec:prelim}
\begin{table}
\small
\caption{\edit{Frequently used notation.}}
\label{tbl:notation}
\centering
\edit{ 
\begin{tabular}{ |c|m{0.6\columnwidth}| } \hline 
\textbf{Notation}&       \textbf{Description} \\ \hline
$G, m, n$&  Social network graph with $n$ nodes and $m$ edges \\ \hline
$F$, $M$&   The misinformation and mitigation campaigns, respectively \\ \hline
$S_F$, $S_M$&   Seed sets for campaigns $F$ and $M$ \\ \hline
$m_F(u,v), m_M(u,v)$&   Meeting event probabilities \\ \hline
$X$&    Possible world of the TCIC model \\ \hline
$h_e^F$, $h_e^M$&   Sampled meeting event edge lengths along edge $e$ in $X$ \\ \hline
$\tau_v$, $\pi_v$&  Sampled AW length and in-neighbour permutation in $X$ \\ \hline
$t_v^F$, $t_v^M$&   First step $v$ meets with $F$, $M$ \\ \hline
$R_F^X$&    Set of nodes reachable from $F$ in $X$ \\ \hline
$\rho_X(v, S)$&    Reward achieved in $X$ at $v$ by $S$ \\ \hline
$\mu_X(S_M)$, $\mu(S_M)$&   Total reward achieved by $S$ \\ \hline
$\overline\mu$, $\underline\mu$&    Submodular upper/lower bounding functions of $\mu$ \\ \hline
$INF_F$, $INF_1$&   Expected influence of $F$ and largest expected influence of any size-1 node set \\ \hline
$\Gamma$&   Misinformation sampling error \\ \hline
$EPT$&  Expected complexity of generating an RDR set \\ \hline
\end{tabular}
}
\end{table} 
We formalize the notions of differential propagation rates and activation windows and then present our new propagation model referred to as \model (for Temporal Competitive Independent Cascade). Let $G=(V,E)$ be a social network with sets of  nodes $V$ and directed edges $E$, where $|V| = n$,  $|E| = m$. Let $F$ (for ``{\em F}ake'') and $M$ (for ``{\em M}itigating'') denote two influence campaigns with seed sets $S_F$ and $S_M$, respectively.\footnote{We assume  $S_F \cap S_M = \emptyset$, w.l.o.g.} The seeds $S_F$ ($S_M$) are active in campaign $F$ (resp. $M$) at time $t=0$. We assume each edge $e=(u,v)$ in $E$ is associated with a propagation probability $p(u,v)  \in (0, 1]$. 


\stitle{Meeting Events.} We associate with each edge $e=(u,v)$ \emph{meeting probabilities} $m_F(u,v), m_M(u,v) \in (0, 1]$. At \emph{any} step, an active node $u$ in campaign $F$ (or $M$) meets any of its currently inactive neighbours $v$ independently with probability $m_F(u,v)$ (or $m_M(u,v)$). A node activated at time $T$  attempts to meet with its inactive out-neighbours at \emph{every} step $t \geq T+1$ until there is a successful meeting. Then, if a meeting event occurs between $u$ and $v$ in $F$ (or $M$) for the \emph{first} time at step $t$, then $u$ is given a single chance to activate $v$ in $F$ (or $M$) with independent probability $p(u,v)$. If the activation attempt is successful, $v$ becomes active in $F$ (or $M$) at time $t$ and enters an activation window, described below. Once activated, subsequent meeting attempts from in-neighbours are ignored.

With the above formulation, we can capture the observation that \emph{truth travels more slowly than fake news} by setting $m_F(u,v) \geq m_M(u,v)$, $\forall$ edge $(u,v)\in E$. As a special case, we can set $m_F(u,v) = 1$, i.e., only apply the meeting events to nodes active in the mitigating campaign $M$. This results in $F$ traversing every edge in a single hop, while $M$ may be delayed by several meeting attempts. For convenience, under this special case, for an edge $e=(u,v)$ we will write $m_M(u,v)$ as $m(u,v)$ or $m(e)$. We shall henceforth assume this special case for ease of exposition, although our theory and techniques apply to the general case.

\stitle{Activation Windows.} Motivated by concepts well established in \eat{other fields such as} Sociology and Marketing \cite{beal1956diffusion, bhagat2012maximizing, kalish1985new, gabielkov2016social, mitchell2016long, lu2015competition}, we distinguish between \emph{awareness} and \emph{adoption} in our \model model by defining an \emph{activation window (AW)} for each node. The AW augments the propagation behaviour defined by meeting events such that successful activation attempts now result in awareness, and not adoption. Specifically, we say a node $u$ becomes \emph{aware} in the \emph{first} timestep that an in-neighbour of $u$ active in either campaign $M$ or $F$ succeeds in meeting with $u$. However, before deciding to commit to activation in either campaign, $u$ enters into its \emph{activation window}: a period of time during which it may receive conflicting and/or reinforcing information, which it factors in making an adoption decision. The dynamics of the AW are governed by a node-level parameter $\gamma(u)$ that determines the length of the AW for node $u$. The \emph{closing} function $\gamma(\cdot) \to \mathbb{Z^{+}}$ allows the model to capture varying window sizes, typically related to the time spent reading articles linked in social media posts. Example choices for the closing function include: (i) a constant function where all nodes wait for some fixed number of steps before making a decision, (ii) a uniform function where window lengths are chosen uniformly at random between $0$ and some closing time $\tau$, or (iii) an attenuating function (e.g.\ exponential decay with an appropriate mapping to $\mathbb{Z^{+}}$) where some users may wait substantially longer (to gather additional information) 
before making an \laks{adoption} decision. Furthermore, $\gamma(\cdot)$ can be made \emph{node-specific} to capture the individual behaviour of users in the network.

\stitle{Tie-breaking Policy.} Observe that (active users from) both campaigns may meet with a node $u$ within the AW and at the end of the window, $u$ must decide to adopt $F$ or $M$. In such a scenario, both campaigns are attempting to activate $u$ and so we require a tie-breaking policy. We employ a \emph{weighted random choice} policy based on in-neighbour activation counts for each campaign over the duration of the AW, described as follows: when the AW closes, the probability that node $u$ activates in campaign $F$ is $| N_F^-(u) | / | N^-(u) |$ where $N_F^-(u)$ (resp. $N^-(u)$) is the set of in-neighbours of $u$ that met with $u$ from campaign $F$ (resp. either campaign). Similarly, the probability that node $u$ activates in campaign $M$ is $| N_M^-(u) | / | N^-(u) |$. 
\eat{Thus, our tie-breaking policy is a weighted random choice based on in-neighbour activation counts for each campaign over the duration of the AW.} 

\stitle{\model Model.} Our new propagation model is defined by incorporating the edge-level meeting events and node-level activation windows into a standard competitive independent cascade model \cite{Chen2013}. The propagation process terminates when all active nodes in both campaigns have met with all their out-neighbours and no new nodes can be activated in either campaign. The model parameters allow for a great deal of expressiveness to accurately describe many observed real-world adoption behaviours.

\stitle{\edit{Problem Statement.}} \edit{We argue that in gauging the extent of misinformation mitigation, merely counting the number of users prevented from adopting the misinformation after the propagation terminates is too restrictive. The reason is that the quality of the mitigation depends on the delay, if any, in the arrival of the truth. To help capture this, assume that there is a reward function $\rho(\cdot): V\times 2^V \rightarrow R^{>=0}$, which given a node $v$ and a seed set $S_M$, returns a real number indicating the effectiveness of the mitigation at $v$ after the propagation terminates. We will provide more details on $\rho(\cdot)$ in the next section. Given such a reward function, we define the \textit{expected mitigation} by $\mu(S_M) = \mathbb{E}[\sum_{v} \rho(v, S_M)]$ where the expectation is taken over the randomness in the propagation process. We are now ready to formally state the problem we study in this work.}
\begin{prob}
\label{prob:mm}
Given a misinformation seed set $S_F$, the \emph{misinformation mitigation (MM) problem} under the \model model is to find a seed set $S_M$ with at most $k$ nodes that maximizes the expected mitigation. Formally, find a seed set $S_M$ satisfying $\argmax_{S_M \subset V, |S_M| \le k} \mu(S_M)$.
\end{prob}

\stitle{Possible Worlds.} We can view the stochastic propagation process under \model using an equivalent ``possible worlds'' interpretation. Suppose that before the propagation process starts, a set of outcomes for all meeting event attempts (i.e.,\ the number of failed meeting attempts before two neighbours successfully meet), activation windows parameters, and edge liveness are pre-determined. Specifically, for each edge $e=(u,v) \in E$, we declare the edge ``live'' with probability $p(u,v)$, or ``blocked'' otherwise. Further, for each edge, we sample from a geometric distribution parameterized by success probability $m_F(e)$ ($m_M(e)$) the (random) number of meeting event attempts, denoted $h_e^F$ ($h_e^M$), required by campaign $F$ (or $M$) along $e$. Next, a set of outcomes for all activation windows are pre-determined. Specifically, for each node $v \in V$ we sample a window length $\tau_v$ from the closing function $\gamma(v)$ and a random permutation $\pi_v$ of the active in-neighbours of $v$ (i.e.,\ those active  in-neighbours connected to $v$ by live edges) for the purpose of tie-breaking. All random events, including coin flips are independent. Therefore, a certain set of outcomes of all coin flips and sampled parameter values corresponds to one possible world under the \model model. A possible world, denoted $X$, is a deterministic graph obtained by conditioning on a particular set of outcomes. We denote by $R_F^X$ the set of nodes reachable by $F$ in possible world $X$ in the absence of $M$.

Next, we define the notion of distance in $X$ for each campaign. Consider a live edge $e = (u,v)$ in $X$. Traditionally, without meeting events, $v$ is reachable from $u$ in a single hop. Now with pre-determined meeting \LL{event attempts}, $v$ is reachable from $u$ in $h_e^F$ (or $h_e^M$) hops. Then, the \emph{delayed-distance} of a path $P$ from $u$ to $v$ is the total number of hops along the edges of $P$ plus the sum of activation window lengths of nodes in $P \setminus \{u,v\}$. \eat{, i.e.,\ the non-terminal nodes in $P$.} Finally, for campaign $A \in \{F, M\}$, the delayed-distance $dd_X(S_A,v)$ from a seed set $S_A$ to $v$ in $X$ is the delayed-distance of the live-path $P$ from $u \in S_A$ to $v$ that minimizes $\sum_{e \in P} h_e^A + \sum_{x \in P \setminus \{u,v\}} \tau_x$.

\begin{figure}
\centering
\resizebox{0.5\linewidth}{!}{
\begin{tikzpicture}
[
  declare function={
    func(\x)= (\x<-4) * (2)   +
     and(\x>=-4, \x<4) * (1)  +
     and(\x>=4, \x<15) * (0);
  }
]
\begin{axis}[
  axis x line=middle, axis y line=middle,
  ymin=-0.25, ymax=2.75, ylabel=$\rho_X(v)$,
  xmin=-9, xmax=9, xlabel=$delay$,
  xtick = {-4,4},
  samples=1000,
  y=0.7cm,
  x=0.4cm
]

\draw [dotted] (4,0) -- (4,2.2);
\draw [dotted] (-4,0) -- (-4,2.2);
\addplot[blue, domain=-10:10]{func(x)};
\end{axis}
\end{tikzpicture}}
\caption{$\rho_X(v_3)$ with $\tau_{v_3} = 4$.}
\label{fig:reward_function}
\end{figure}
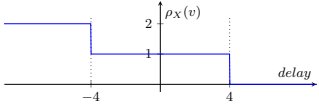

\begin{example}
{\em \edit{We illustrate the propagation process on the network in Figure \ref{fig:intro} under the \model model. For simplicity, we assume each edge has propagation probability $p(e) = 1$. The sampled meeting lengths $h_e^M$ and AW lengths $\tau_v$ are indicated in \cyan{cyan} and \red{red} labels respectively. Let $\pi_{v_4} = (v_{11}, v_{10}, v_F)$ and $\pi_{v_{15}} = (v_5, v_{12})$ be the sampled permutations. Consider seeds $v_0$ for campaign $F$ and $v_{12}$ for campaign $M$ and consider the resulting cascade through the network. At time 1, $v_1, v_2, v_5, v_6$ adopt $F$ and $v_{11}$ adopts $M$. Meanwhile, $v_3, v_4$ open their AW’s. At time 2, both $F$ and $M$ reach $v_{15}$. Additionally, $M$ reaches $v_4$ as its AW closes. Tie-breaks result in $v_4$ and $v_{15}$ adopting $M$ and $F$ respectively. Meanwhile $v_7, v_8, v_9$ adopt $F$ and $v_{14}$ adopts $M$. Finally, at time 5 $v_{10}$ adopts $M$ and $v_3$ adopts $F$.}  

\edit{Using the above sampled parameters, with $v_0$ as the $F$ seed, let us compare $M$ seeds $v_{12}$ and $v_{14}$. It can be verified that $v_{12}$ causes the nodes $v_{12}, v_{11}, v_4, v_{14}, v_{10}$ to either adopt or be informed of $M$. Of these only $v_{11}$ and $v_3$ would have adopted $F$ if there was no $M$ campaign, so intuitively $v_{12}$ ``saves" 2 nodes. By contrast, $v_{14}$ causes $v_{14}, v_{10}, v_3$ to adopt or be informed of $M$, of which only $v_3$ would have adopted $F$ if there was no $M$ campaign, so $v_{14}$ ``saves" $v_3$. \qed 
}
}
\end{example}

\subsection{Mitigation Reward}
\label{sec:prob_defn}

Motivated by the time-critical nature of the MM problem, we introduce a novel reward function that intuitively captures the penalty paid when the mitigation campaign arrives too late after the misinformation. In particular, the reward function is designed such that in \LL{case} adoption of the truth cannot be secured, awareness is encouraged.

\stitle{Delay-specific Reward Function.} The reward function $\rho_X(\cdot)$ is defined w.r.t.\ activating a node in campaign $M$ relative to the behaviour of campaign $F$. First, for nodes $v \notin R_F^X$ that would \emph{not} have been activated in $F$ in the absence of $M$ we define $\rho_X(v, S_M) = 0$. Next, consider a node $v$ that would have been activated in $F$ in the absence of $M$, i.e.\ there exists a path from $S_F$ to $v$ in $X$. Let $t_v^F$ be the first step in which $v$ meets with a node in campaign $F$ and $t_v^M$ be the first step in which $v$ meets with a node in campaign $M$. We use the convention that $t_v^F$ (or $t_v^M$) is infinite if no meeting with a node from $F$ (or $M$) occurs. We define the reward $\rho_X(v, S_M)$ for node $v$ as a function of the amount of time that has passed between the mitigation and the fake news arriving at $v$. In particular, we consider the step function given by Eq. (\ref{eqn:reward}). There are three cases for the amount of reward achieved: (i) reward $2$ if the truth arrives at $v$ sufficiently early such that the misinformation arrives after the AW closes or if the presence of $M$ stops $F$ from ever reaching $v$ (in the case of $t_v^F = \infty$), (ii) reward $1$ if both the truth and misinformation arrive at $v$ within the AW and (iii) no reward if the truth arrives after the AW closes or not at all (i.e.,\ $t_v^M = \infty$). The reward function is illustrated in Figure~\ref{fig:reward_function} for $v_3$ from Figure~\ref{fig:intro} where $\tau_{v_3} = 4$.
\begin{equation}
\label{eqn:reward}
\rho_X(v, S_M)=
\begin{cases}
2 &\mbox{if } t_v^M < t_v^F - \tau_v \\
1 &\mbox{if } | t_v^M - t_v^F | \leq \tau_v \\
0 &\mbox{if } t_v^M > t_v^F + \tau_v \\
\end{cases}
\end{equation}
When the context is clear, we  write $\rho_X(v, S_M)$ as $\rho(v)$. We refer to the reward achieved by set $S_M$, given $S_F$, after the propagation terminates as the \emph{mitigation} and denote it by $\mu_X(S_M) = \sum_{v} \rho_X(v)$. Further, we denote the \textit{expected mitigation} by $\mu(S_M) = \mathbb{E}[\mu_X(S_M)]$.\footnote{Mitigation depends on $S_F$, but we omit $S_F$ as an argument of $\mu(.)$ since $S_F$ is a fixed input to the problem.}


\stitle{Design Decisions.} The advantage of our new reward function, and the purpose of the middle condition of Equation \ref{eqn:reward}, is to promote campaign $M$ reaching nodes that are reached by campaign $F$, despite not being able to {\sl guarantee} adoption of $M$, \LL{owing to tie-breaking}. Thus, even if a node does not end up activating in $M$, the user will be exposed to the true information, which we argue is a natural goal to strive for. Clearly, in the event that adoption \LL{of $M$} cannot be guaranteed, promoting \LL{its} awareness is preferable to inaction. 
If the truth reaches a user sufficiently early (compared to the misinformation) then their adoption decision will be uncontested by the misinformation. Thus, the subsequent propagation of truth by the user leads to the desired outcome. Hence, this favourable scenario contributes maximal reward towards the objective. Meanwhile, due to the tie-breaking policy, reaching users within their AW always provides an opportunity for the adoption of truth and its subsequent propagation. However, when both the misinformation and truth arrive within the AW, the adoption of truth is no longer uncontested and, as such, contributes less reward. Finally, when mitigation arrives too late after the misinformation, we penalize the mitigation campaign to capture its reduced effectiveness. Further, recall that the reward function only attributes non-zero values when campaign $F$ reaches node $v$. Since we do not credit any reward to those nodes reached by the mitigation but not the misinformation, this encourages a solution to ``focus'' on mitigating the spread of misinformation, and discourages solutions that blindly maximize the spread of the truth. 
\edit{We note that in place of a step function, we could use an arbitrary non-increasing reward function. We revisit this point in \S~\ref{sec:sandwich}.} 

\stitle{Problem Properties.} The original MM problem under the CIC model \cite{budak2011limiting} is \textbf{NP}-hard. It is a special case of MM under the \model model, with all $m(u,v)=1$, a length zero activation window, and dominant tie-breaking.
\begin{prop}
The misinformation mitigation problem is \textbf{NP}-hard under the \model model.
\end{prop}
An important observation is that, while the expected mitigation $\mu(\cdot)$ is monotone under the \model model with reward function $\rho(\cdot)$, it is not submodular.
\begin{thm}
\label{thm:negative}
Given a seed set $S_F$, the mitigation function $\mu(\cdot)$ is not submodular in general under the \model model.
\end{thm}

\begin{proof}
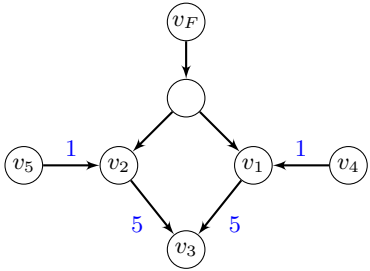
\begin{figure}
\centering
\begin{tikzpicture}

\tikzset{node/.style={circle,draw,minimum size=0.5cm,inner sep=0pt},}
\tikzset{edge/.style={->,> = latex'},}

\node[node] (1) {$v_F$};
\node[node] (2) [below = 0.5cm of 1] {};
\node[node] (3) [below right = 0.75cm of 2] {$v_1$};
\node[node] (4) [below left = 0.75cm of 2] {$v_2$};
\node[node] (5) [below = 1.5cm of 2] {$v_3$};
\node[node] (6) [right = 0.75cm of 3] {$v_4$};
\node[node] (7) [left = 0.75cm of 4] {$v_5$};

\path[draw,thick,->,> = latex']
(1) edge node {} (2)
(2) edge node {} (3)
(2) edge node {} (4)
(3) edge node[below right, blue] {5} (5)
(4) edge node[below left, blue] {5} (5)
(6) edge node[above, blue] {1} (3)
(7) edge node[above, blue] {1} (4);

\end{tikzpicture}
\caption{Diamond counter-example.}\label{fig:thm_one}
\end{figure}

Consider the graph shown in Figure \ref{fig:thm_one} and let $S_F = \{ v_F \}$. Assume that all edges probabilities $p(e) = 1$ and all activation window lengths are zero. Consider a possible world $X$ in which the meeting lengths for campaign $M$ are as shown in blue. Now, consider the sets $A = \emptyset$, $B = \{ v_5 \}$ and node $v_4$. Then, it is straightforward to compute the mitigation $\mu_X(A) = 0$ and $\mu_X(B) = 2$. In particular, when $B$ is used as the seed set for $M$, node $v_2$ ends up adopting $M$, but due to the meeting length of the edge $(v_2,v_3)$, campaign $M$ reaches $v_3$ after its activation window has closed, resulting in no reward at $v_3$ according to $\rho_X(\cdot)$. Next, consider the mitigation achieved when adding $v_4$ to $A$ and $B$. We have, $\mu_X(A \cup \{ v_4 \}) = 2$. On the other hand, $\mu_X(B \cup \{ v_4 \}) = 6$. Thus, since $\mu_X(B \cup \{ v_4 \}) - \mu_X(B) > \mu_X(A \cup \{ v_4 \}) - \mu_X(A)$ and $A \subseteq B$, the mitigation objective does not satisfy submodularity.
\end{proof}

Bian et al.\ \cite{bian2017guarantees} recently provided approximation guarantees for the standard greedy algorithm on non-submodular objectives \LL{based} on the submodularity ratio and curvature of the objective. Unfortunately, the submodularity ratio of the MM objective can be as small as $0$ and thus their result does not provide any non-trivial guarantees. To overcome the challenges of a non-submodular objective function, we leverage the \emph{Sandwich Approximation} of \cite{lu2015competition} by developing appropriate upper and lower bounding functions for our mitigation objective, yielding a solution with \LL{non-trivial} data-dependent approximation guarantees.
\section{Sandwiching the Mitigation Objective}
\label{sec:sandwich}

The \emph{Sandwich Approximation (SA)} technique of \cite{lu2015competition} leverages upper ($\overline\mu$) and lower ($\underline\mu$) bounds on a non-submodular objective function to provide data-dependent approximation guarantees. Specifically, given upper and lower bounding functions that are themselves submodular, we can obtain solutions $S^U$ and $S^L$ resp.\ with approximation guarantees.

\stitle{Lower Bound.} We make the following observation to motivate our choice of lower bound objective $\underline\mu(\cdot)$. The supermodular behaviour of $\mu(\cdot)$ arises from the combined effort of seed nodes in $S_M$ which individually would not activate a target node $v$. That is, together they are able to block the paths from $S_F$ to $v$ such that the mitigating campaign's disadvantage, in the form of meeting events, is overcome. \edit{E.g., in Figure~\ref{fig:intro}, suppose $S_F = \{v_1\}$. Then neither one of $v_2, v_7$ by itself can activate $v_8$ in $M$. However, $S_M = \{v_2, v_7\}$ can.} Thus, to eliminate the possibility of such \emph{coordination}, we define a lower bound function $\underline\mu_X$ that only measures the maximum mitigation achieved by any node in $S_M$ when it acts as a singleton seed set. Formally, we define $\underline\mu_X(S_M) = \sum_{v} \max_{u \in S_M} \rho_X(v,\{ u \}) $ and $\underline\mu(S_M) = \mathbb{E}[\underline\mu_X(S_M)]$. Note, $\underline\mu_X$ clearly lower bounds $\mu_X$ since $\max_{u \in S_M} \rho_X(v,\{ u \}) \le \rho_X(v,S_M)$.
\begin{lem} \label{lem:lbsm} 
$\underline\mu(S_M)$ is submodular.
\end{lem}

\begin{proof}
Fix a possible world $X$. The set of nodes that activate in campaign $F$ when $S_M = \emptyset$ is the set of nodes reachable from $S_F$ in $X$, denoted $R_F^X$. For any node $v \in R_F^X$, when a node is added to $S_M$, $v$ will either continue to activate in $F$ or switch to activating in $M$. To prove submodularity it is sufficient to prove submodularity for each $v \in R_F^X$. Denote ${\underline m}_v(S) = \max_{u \in S} \rho_X(v,\{ u \})$ as the mitigation at node $v$ due to seed set $S$. Denote $t_v^F$ as the first step in which $v$ meets with a node in campaign $F$ in possible world $X$ and $t_v^M$ is the first step in which $v$ meets with a node in campaign $M$ in possible world $X$. Define $\Delta_v(S) = t_v^M - t_v^F$ as the delay at node $v$ under mitigating seed set $S$ in possible world $X$.

Consider any sets $A$ and $B$ such that $A \subseteq B$. Let $w \in V \setminus B$ be an arbitrary node. Consider a node $v \in R_F^X$ such that

\vspace{-10pt}

\begin{equation}
\label{eqn:sm_one}
    {\underline m}_v(B \cup \{ w \}) > {\underline m}_v(B).
\end{equation}

In order to satisfy (\ref{eqn:sm_one}), the mitigation at $v$ must increase with the inclusion of node $w$. Observe that due to the form of the reward function $\rho_X(\cdot)$, such an increase can only occur at the ``steps'' along $\rho_X(\cdot)$ (see Figure~\ref{fig:prelim}(b)). We consider the following two cases below: (i) $v$ is reached by $F$ under seed set $B \cup \{ w \}$ and (ii) $v$ is not reached by $F$ under seed set $B \cup \{ w \}$.

\emph{\underline{Case (i)}:} First, consider the case where $v$ is reached by $F$ under both seed sets $B$ and $B \cup \{ w \}$ in $X$. In this case, in order for the mitigation at $v$ to increase by crossing a ``step'' in $\rho_X(\cdot)$, we must have that $\min_{u \in B \cup \{ w \}} \Delta_v(\{ u \}) < \min_{u \in B} \Delta_v(\{ u \})$. Then, consider an $M$-active path $P^M$ from some node in $B \cup \{ w \}$ to $v$ that results in delay \linebreak $\min_{u \in B \cup \{ w \}} \Delta_v(\{ u \})$. If the starting node of $P^M$ is not $w$, then we know that $\min_{u \in B \cup \{ w \}} \Delta_v(\{ u \}) = \min_{u \in B} \Delta_v(\{ u \})$ which contradicts (\ref{eqn:sm_one}). Therefore, the path responsible for the increase in mitigation at $v$ must start from $w$. As a result, $\Delta_v(\{ w \}) < \min_{u \in B} \Delta_v(\{ u \})$.

\emph{\underline{Case (ii)}:} Second, consider the case where $v$ adopts $M$ without resolving a tie-break for seed set $B \cup \{ w \}$. Then, in order to adhere to (\ref{eqn:sm_one}), it must be the case that under seed set $B$, $v$ was reached by campaign $F$ and the reward achieved was at most $1$. In particular, it must be the case that $\min_{u \in B} \Delta_v(\{ u \}) \geq - \tau_v$. Now, if $v$ adopts $M$ without resolving a tie-break under seed set $B \cup \{ w \}$, then all predecessors of $v$ must be $M$-active. As a result, since each $u \in B \cup \{ w \}$ is run as a singleton seed set for $\underline\mu_X(\cdot)$, it must be the case that there is an $M$-active path from $w$ to each predecessor of $v$.

Now, for set $A \subseteq B$, for case (i) we have from the above that $\Delta_v(\{ w \}) < \min_{u \in B} \Delta_v(\{ u \}) \leq \min_{u \in A} \Delta_v(\{ u \})$ since delay is monotonically decreasing w.r.t.\ the size of the mitigating seed set. For case (ii) we similarly have that there is an $M$-active path from $w$ to each predecessor of $v$ while for each singleton seed set from $A$ we have $\min_{u \in A} \Delta_v(\{ u \}) \geq \min_{u \in B} \Delta_v(\{ u \}) \geq - \tau_v$. It follows that ${\underline m}_v(A \cup \{ w \}) > {\underline m}_v(A)$.

Finally, since each node is run as a singleton, we have that node $w$ is solely responsible for the increase in reward achieved at $v$. As a consequence, we have ${\underline m}_v(A \cup \{ w \}) = {\underline m}_v(B \cup \{ w \}) = {\underline m}_v(\{ w \})$. Thus, from the monotonicity of ${\underline m}(\cdot)$, we have that ${\underline m}_v(A \cup \{ w \}) - {\underline m}_v(A) \geq {\underline m}_v(B \cup \{ w \}) - {\underline m}_v(B)$ as required.
\end{proof}

\stitle{Upper Bound.} An obvious candidate for $\overline\mu(\cdot)$ is to forego the meeting events associated with campaign $M$ and enforce an $M$-dominant tie-breaking rule. The resulting model reduces to the CIC model under which previous results ensure that the resulting objective is submodular. The existence of meeting events only acts to ``hinder'' the mitigation and thus without them the mitigation would reach every node sooner, thus increasing reward. Further, $M$-dominant tie-breaking rule ensures all propagation paths shared by the two campaigns are won by the mitigation. However, we develop a tighter upper bounding function so as  to improve the data-dependent approximation guarantees.

Consider a possible world $X$ and call all edges touched by the propagation of campaign $F$ in $X$, \emph{critical} edges $E_C \subseteq E$. Next, construct a modified possible world $X'$ by removing all meeting events for campaign $M$ on critical edges. That is, $h_e^M = 1$ for all critical edges $e \in E_C$. Next, define an \emph{overlap} indicator variable $\mathbb{I}_{v}^{OL}$ where $\mathbb{I}_{v}^{OL} = 0$ iff the collection of paths from $S_F$ to $v$ is \LL{edge-disjoint} from the collection of paths from $S_M$ to $v$ in $X'$. Intuitively, if there is no overlap then there is no opportunity for several seeds from $S_M$ to ``collude'' together to achieve a larger reward than if they were to act alone. Finally, we define a modified reward function $\rho'_X(\cdot)$ that upper bounds $\rho_X(\cdot)$ by lifting the reward to its maximum value for the case that both campaigns reach some node $u$ within its activation window. Denote $t_v^{A,X}$ as the first step in which $v$ meets with a node in campaign $A$ in possible world $X$, where $A \in \{F, M\}$. We similarly define $\rho'_X(v,S_M) = 0$ when $v \notin R_F^X$, and otherwise as
\begin{equation}
\label{eqn:reward_prime}
\rho'_X(v, S_M)=
\begin{cases}
2 &\mbox{if } t_v^{M,X} \leq t_v^{F,X} + \tau_v \\
0 &\mbox{if } t_v^{M,X} > t_v^{F,X} + \tau_v.
\end{cases}
\end{equation}
Then, we define
\begin{equation}
\label{eqn:upper_mit}
\overline\mu_X(S_M) = \sum_v \max_{u \in S_M}
\begin{cases}
\rho_{X'}(v,\{ u \}) &\mbox{if } \mathbb{I}_{v}^{OL} = 0 \\
\rho'_{X'}(v,\{ u \}) &\mbox{if } \mathbb{I}_{v}^{OL} = 1
\end{cases}
\end{equation}
and $\overline\mu(S_M) = \mathbb{E}[\overline\mu_X(S_M)]$. 
\begin{lem}
$\overline\mu(S_M)$ is submodular.
\end{lem}
\eat{ For lack of space, we omit the proof. 
The proof technique follows the same logic as the proof of Lemma~\ref{lem:lbsm}, with some additional case analysis to account for the form of Eq. \ref{eqn:upper_mit}. 
} 

\begin{proof}
Fix possible worlds $X$ and $X'$. The set of nodes that activate in campaign $F$ when $S_M = \emptyset$ is the set of nodes reachable from $S_F$ in $X$, denoted $R_F^X$. For any node $v \in R_F^X$, when a node is added to $S_M$, $v$ will either continue to activate in $F$ or switch to activating in $M$. To prove submodularity it is sufficient to prove submodularity for each $v \in R_F^X$. Define the delay at node $v$ in possible world $X$ under mitigating seed set $S$ as $\Delta_{v,X}(S) = t_v^{M,X} - t_v^{F,X}$ and the mitigation at node $v$ by

\vspace{-5pt}

\begin{equation}
{\overline m}_v(S)=
\begin{cases}
\max_{u \in S} \rho_{X'}(v,\{ u \}) &\mbox{if } \mathbb{I}_{v}^{OL} = 0 \\
\max_{u \in S} \rho'_{X'}(v,\{ u \}) &\mbox{else} \\
\end{cases}
\end{equation}

Consider any sets $A$ and $B$ such that $A \subseteq B$. Let $w \in V \setminus B$ be an arbitrary node. Consider a node $v \in R_F^X$ such that ${\overline m}_v(B \cup \{ w \}) > {\overline m}_v(B)$.

\emph{\underline{Case $\mathbb{I}_{v}^{OL} = 0$}:} In order to satisfy (\ref{eqn:sm_one}), the mitigation at $v$ must increase with the inclusion of node $w$. Observe that due to the form of the reward function $\rho_X(\cdot)$, such an increase can only occur at the ``steps'' along $\rho_X(\cdot)$. We consider the following two cases below: (i) $v$ is reached by $F$ under seed set $B \cup \{ w \}$ and (ii) $v$ is not reached by $F$ under seed set $B \cup \{ w \}$.

\emph{\underline{Case (i)}:} First, consider the case where $v$ is reached by $F$ under both seed sets $B$ and $B \cup \{ w \}$. In this case, in order for the mitigation at $v$ to increase by crossing a ``step'' in $\rho_X(\cdot)$, we must have that $\min_{u \in B \cup \{ w \}} \Delta_{v,X'}(\{ u \}) < \min_{u \in B} \Delta_{v,X'}(\{ u \})$. Then, consider an $M$-active path $P^M$ from some node in $B \cup \{ w \}$ to $v$ that results in delay $\min_{u \in B \cup \{ w \}} \Delta_{v,X'}(\{ u \})$. If the starting node of $P^M$ is not $w$, then we know that $\min_{u \in B \cup \{ w \}} \Delta_{v,X'}(\{ u \}) = \min_{u \in B} \Delta_{v,X'}(\{ u \})$ which contradicts (\ref{eqn:sm_one}). Therefore, the path responsible for the increase in mitigation at $v$ must start from $w$. As a result, $\Delta_{v,X'}(\{ w \}) < \min_{u \in B} \Delta_{v,X'}(\{ u \})$.

\emph{\underline{Case (ii)}:} Second, consider the case where $v$ adopts $M$ without resolving a tie-break for seed set $B \cup \{ w \}$. Then, in order to adhere to (\ref{eqn:sm_one}), it must be the case that under seed set $B$, $v$ was reached by campaign $F$ and the reward achieved was at most $1$. In particular, it must be the case that $\min_{u \in B} \Delta_{v,X'}(\{ u \}) \geq - \tau_v$. Now, if $v$ adopts $M$ without resolving a tie-break under seed set $B \cup \{ w \}$, then all predecessors of $v$ must be $M$-active. As a result, since each $u \in B \cup \{ w \}$ is run as a singleton seed set for $\overline\mu_X(\cdot)$, it must be the case that there is an $M$-active path from $w$ to each predecessor of $v$.

Now, for set $A \subseteq B$, for case (i) we have from the above that $\Delta_{v,X'}(\{ w \}) < \min_{u \in B} \Delta_{v,X'}(\{ u \}) \leq \min_{u \in A} \Delta_{v,X'}(\{ u \})$ due to the monotonicity of delay. For case (ii) we similarly have that there is an $M$-active path from $w$ to each predecessor of $v$ while for each singleton seed set from $A$ we have $\min_{u \in A} \Delta_{v,X'}(\{ u \}) \geq \min_{u \in B} \Delta_{v,X'}(\{ u \}) \geq - \tau_v$. It follows that ${\overline m}_v(A \cup \{ w \}) > {\overline m}_v(A)$. Finally, since each node is run as a singleton, we have that node $w$ is solely responsible for the increase in reward achieved at $v$. As a consequence, we have ${\overline m}_v(A \cup \{ w \}) = {\overline m}_v(B \cup \{ w \}) = {\overline m}_v(\{ w \})$. Thus, from the monotonicity of ${\overline m}(\cdot)$, we have that ${\overline m}_v(A \cup \{ w \}) - {\overline m}_v(A) \geq {\overline m}_v(B \cup \{ w \}) - {\overline m}_v(B)$ as required.


\emph{\underline{Case $\mathbb{I}_{v}^{OL} = 1$}:} In order to satisfy (\ref{eqn:sm_one}), the mitigation at $v$ must increase with the inclusion of node $w$. Observe that due to the form of the reward function $\rho'_X(\cdot)$, such an increase can only occur at the ``step'' along $\rho'_X(\cdot)$. We consider the following two cases below: (i) $v$ is reached by $F$ under seed set $B \cup \{ w \}$ and (ii) $v$ is not reached by $F$ under seed set $B \cup \{ w \}$.

\emph{\underline{Case (i)}:} First, consider the case where $v$ is reached by $F$ under both seed sets $B$ and $B \cup \{ w \}$. In this case, in order for the mitigation at $v$ to increase by crossing the ``step'' in $\rho'_X(\cdot)$, we must have that $\min_{u \in B \cup \{ w \}} \Delta_{v,X'}(\{ u \}) < \min_{u \in B} \Delta_{v,X'}(\{ u \})$. Then, consider an $M$-active path $P^M$ from some node in $B \cup \{ w \}$ to $v$ that results in delay $\min_{u \in B \cup \{ w \}} \Delta_{v,X'}(\{ u \})$. If the starting node of $P^M$ is not $w$, then we know that $\min_{u \in B \cup \{ w \}} \Delta_{v,X'}(\{ u \}) = \min_{u \in B} \Delta_{v,X'}(\{ u \})$ which contradicts (\ref{eqn:sm_one}). Therefore, the path responsible for the increase in mitigation at $v$ must start from $w$. As a result, $\Delta_{v,X'}(\{ w \}) < \min_{u \in B} \Delta_{v,X'}(\{ u \})$.

\emph{\underline{Case (ii)}:} Second, consider the case where $v$ adopts $M$ without resolving a tie-break for seed set $B \cup \{ w \}$. Then, in order to adhere to (\ref{eqn:sm_one}), it must be the case that under seed set $B$, $v$ was reached by campaign $F$ and the reward achieved was at most $0$. In particular, it must be the case that $\min_{u \in B} \Delta_{v,X'}(\{ u \}) > \tau_v$. Now, if $v$ adopts $M$ without resolving a tie-break under seed set $B \cup \{ w \}$, then all predecessors of $v$ must be $M$-active. As a result, since each $u \in B \cup \{ w \}$ is run as a singleton seed set for $\overline\mu_X(\cdot)$, it must be the case that there is an $M$-active path from $w$ to each predecessor of $v$.

Now, for set $A \subseteq B$, for case (i) we have from the above that $\Delta_{v,X'}(\{ w \}) < \min_{u \in B} \Delta_{v,X'}(\{ u \}) \leq \min_{u \in A} \Delta_{v,X'}(\{ u \})$ due to the monotonicity of delay. For case (ii) we similarly have that there is an $M$-active path from $w$ to each predecessor of $v$ while for each singleton seed set from $A$ we have $\min_{u \in A} \Delta_{v,X'}(\{ u \}) \geq \min_{u \in B} \Delta_{v,X'}(\{ u \}) > \tau_v$. It follows that ${\overline m}_v(A \cup \{ w \}) > {\overline m}_v(A)$. Finally, since each node is run as a singleton, we have that node $w$ is solely responsible for the increase in reward achieved at $v$. As a consequence, we have ${\overline m}_v(A \cup \{ w \}) = {\overline m}_v(B \cup \{ w \}) = {\overline m}_v(\{ w \})$. Thus, from the monotonicity of ${\overline m}(\cdot)$, we have that ${\overline m}_v(A \cup \{ w \}) - {\overline m}_v(A) \geq {\overline m}_v(B \cup \{ w \}) - {\overline m}_v(B)$ as required.
\end{proof}

\stitle{Reverse Delayed Reward Sets.} State-of-the-art solutions for the IM problem are based on the concept of \emph{Reverse Reachable (RR)} sets. We denote a possible world under the IC model by $W$.
\begin{defn}[Reverse Reachable Set \cite{tang2015influence}] 
\label{defn:rr_set}
The reverse reachable (RR) set for a root node $v$ in $W$ is the set of nodes that can reach $v$ in $W$. That is, for each node $u$ in the RR set, there is a directed path from $u$ to $v$ in $W$. 
\end{defn}
\begin{defn}[Random RR Set \cite{tang2015influence}]
A random RR set is an RR set generated on a random instance of $W$, for a root node selected uniformly at random from $W$.
\end{defn}
The connection between RR sets and a target node's  activation is formalized in the following equivalence lemma. It is the key result that underpins the approximation guarantees for the approaches of \cite{tang2015influence, nguyen2016stop, tang2018online} by establishing an unbiased estimator for the influence objective.
\begin{lem}[Activation Equivalence \cite{borgs2014maximizing}]
\label{lem:borgs}
For any seed set $S$ and node $v$, the probability that an influence propagation process from $S$ can activate $v$ equals the probability that $S$ overlaps an RR set for $v$.
\end{lem}
For our objective function, we seek an analog to the RR set definition in the \model setting. Importantly, since RR sets are only concerned with the coverage status of a node $u$ w.r.t.   an RR set $R$, the presence of $u$ in $R$ implies this condition is satisfied. By contrast, our analog must be able to express the \textit{reward} associated with each node present in the set. To overcome this challenge, we introduce the notion of a \emph{Reverse Delayed Reward (RDR)} set for a node $v$ \LL{in a possible world $X$ of our TCIC model.} RDR sets augment the traditional RR sets by including \emph{delayed reward} information associated with each node in the RDR set. They can be viewed as a \emph{weighted} version of RR sets.
\begin{defn}[Reverse Delayed Reward Set]
\label{defn:rdr}
The reverse delayed reward (RDR) set for a node $v$ in $X$ is the set of pairs $(u,\rho_X(v,\{u\}))$ of nodes that can reach $v$ in $X$ and their associated reward (which we interpret as a weight). For each node $u$ in the RDR set, there exists a path $P$ in $X$ from $u$ to $v$ for which all tie-breaks along $P$ are won by $M$ when $\{ u \}$ is initially activated in campaign $M$, achieving reward $\rho_X(v,\{u\})$, where $\rho_X(v,\{u\})$ is the reward computed in the possible world $X$. 
\end{defn}
Note that RDR sets are defined w.r.t.\ a fixed seed set $S_F$ given by a problem instance. \edit{E.g., in Figure~\ref{fig:intro}, let all propagation probabilities be $1$ and let $S_F = \{v_0\}$. Then the RDR set for node $v_3$ in this possible world is $\{(v_{10}, 1), (v_{14}, 1), (v_2, 1)\}$. An $M$ campaign started at any other node either does not reach $v_3$ or reaches it too late.} We make the following important observation: the delayed-distance $dd_X(u,v)$ from $u$ to $v$ in $X$ is necessary, but not sufficient, information for determining the delay of $M$ reaching $v$ w.r.t.\ $F$. For example, the simultaneous propagation of $F$ and $M$ in $X$ can lead to interactions resulting in a node's AW opening before either campaign would have triggered it when propagating independently. As a result, this may lead to a delay value that is not simply $dd_X(u,v) - dd_X(S_F,v)$. Therefore, we cannot compute the reward associated with node $u$ directly from the delayed-distance.

\eat{The accompanying notion of} A random RDR set is defined in a similar fashion to a random RR set where the ``root'' node $v$ is chosen at random from $G$. Since the mitigation objective is not submodular (Theorem \ref{thm:negative}), we cannot apply RDR sets directly for maximizing $\mu(.)$. Instead, we will establish a connection between RDR sets and the bounding functions $\overline\mu(\cdot)$ and $\underline\mu(\cdot)$. 
 

\stitle{Reward Equivalence.} \LL{We say a set $S$ \textit{covers} an RDR set $R$ with weight $\omega_{R} > 0$, \eat{abusing notation we denote this as $S \cap R = \omega_R$,} provided  $\exists u\in S$ such that the pair  $(u,\omega_{R})$ appears in $R$ and $\omega_{R}$ is the largest weight over all nodes $u \in S$.  Abusing notation, we write this as $S \cap R = \omega_R$. If there are no pairs $(u,\omega_{R})$ in $R$ such that $u \in S$,  then we define $\omega_{R} = 0$ and say $S$ \textit{does not cover} $R$.} \LL{Reward equivalence for the two bounds are established in Lemmas~\ref{lem:equiv_lower} and \ref{lem:equiv_upper}.}
\begin{lem}
\label{lem:equiv_lower}
Let $S_M$ be a fixed set of nodes, and $v$ be a fixed node. Suppose that we generate an RDR set $\underline R$ for $v$ in a possible world $X$. Let $\varrho_1$ be the probability that $S_M$ covers $\underline R$ with weight $\omega_{\underline R}$, and $\varrho_2$ be the probability that $S_M$, when used as a seed set for campaign $M$, achieves a reward $\omega_{\underline R}$ at $v$ in a propagation process on $G$ w.r.t\ $\underline\mu(\cdot)$. Then, $\varrho_1 = \varrho_2$.
\end{lem}

\begin{proof}
Let $X$ be the possible world sampled from $G$. Then, $\varrho_2$ equals the probability that $v$ is reached from $S_M$ in $X$ and the resulting reward is $\omega_{\underline R}$. In other words, there exists a path $P$ in $X$ from $u \in S_M$ to $v$ for which all tie-breaks along $P$ are won by $M$, i.e.,\ $M$-active. Note, while several nodes $u \in S_M$ may reach $v$, the reward $\omega_{\underline R}$ is computed with respect to the first arrival at $v$. Thus, the probability of achieving reward 2 is equal to the probability there exists an $M$-active path in $X$ with delayed-distance less than $dd(S_F,v) - \tau_v$ from some $u \in S_M$ to $v$. Further, the probability of achieving reward 1 is equal to the probability there exists an $M$-active path in $X$ with delayed-distance in the range $[dd(S_F,v) - \tau_v, dd(S_F,v) + \tau_v]$ and there are no $M$-active paths with delayed-distance less than $dd(S_F,v) - \tau_v$ from \LL{any} $u \in S_M$ to $v$.

Meanwhile, by Definition \ref{defn:rdr}, $\varrho_1$ equals the probability that $X$ contains a path $P$ that ends at $v$ and starts at some node $u \in S_M$ for which all tie-breaks along $P$ are won by $M$ such that $\rho_X(v, \{u\}) = \omega_{\underline R}$ and $\rho_X(v, \{u'\}) \leq \omega_{\underline R}$ for all $u' \neq u \in S_M$. Here, the reward $\omega_{\underline R} = 2$ if there exists an $M$-active path in $X$ with delayed-distance less than $dd(S_F,v) - \tau_v$ from some $u \in S_M$ to $v$. On the other hand, $\omega_{\underline R} = 1$ if there exists an $M$-active path in $X$ with delayed-distance in the range $[dd(S_F,v) - \tau_v, dd(S_F,v) + \tau_v]$ and there are no $M$-active paths with delayed-distance less than $dd(S_F,v) - \tau_v$ from \LL{any}  $u \in S_M$ to $v$. It follows that $\varrho_1 = \varrho_2$.
\end{proof}

\begin{lem} 
\label{lem:equiv_upper}
Let $S_M$ be a fixed set of nodes, and $v$ be a fixed node. Suppose that we generate an RDR set $\overline R$ for $v$ in a possible world $X'$ where $X'$ is the modified possible world constructed from possible world $X$ sampled from $G$. Let $\varrho_1$ be the probability that $S_M$ covers $\overline R$ with weight $\omega_{\overline R}$, and $\varrho_2$ be the probability that $S_M$, when used as a seed set for campaign $M$, achieves a reward $\omega_{\overline R}$ at $v$ in a propagation process on $G$ w.r.t.\ $\overline\mu(\cdot)$. Then, $\varrho_1 = \varrho_2$.
\end{lem}


\begin{proof}
Let $X'$ be the modified possible world constructed from possible world $X$ sampled from $G$. Then, $\varrho_2$ equals the probability that $v$ is reached from $S_M$ in $X'$ and the resulting reward is $\omega_{\overline R}$. In other words, there exists a path $P$ in $X'$ from some $u \in S_M$ to $v$ for which all tie-breaks along $P$ are won by $M$, i.e.\ $M$-active. Note, while several nodes $u \in S_M$ may reach $v$, the reward $\omega_{\overline R}$ is computed with respect to the first arrival at $v$. Now, notice that the reward achieved by each $u$ depends on the overlap indicator. Thus, the probability of achieving reward 2 is equal to the sum of the following two probabilities: (i) the probability that there exists an $M$-active path in $X'$ with delayed-distance less than $dd(S_F,v) - \tau_v$ from some $u \in S_M$ to $v$ that does not overlap with any path from $S_F$ to $v$ and (ii) the probability there exists an $M$-active path in $X'$ with delayed-distance less than $dd(S_F,v) + \tau_v$ from some $u \in S_M$ to $v$ that overlaps with some path from $S_F$ to $v$. Further, the probability of achieving reward 1 is equal to the probability there exists an $M$-active path in $X'$ with delayed-distance in the range $[dd(S_F,v) - \tau_v, dd(S_F,v) + \tau_v]$ and there are no $M$-active paths with delayed-distance less than $dd(S_F,v) - \tau_v$ from some $u \in S_M$ to $v$ that does not overlap with any path from $S_F$ to $v$. 

Meanwhile, by Definition \ref{defn:rdr}, $\varrho_1$ equals the probability that $X'$ contains a path $P$ that ends at $v$ and starts at some node $u \in S_M$ for which all tie-breaks along $P$ are won by $M$ such that $\rho_X(v, \{u\}) = \omega_{\overline R}$ and $\rho_X(v, \{u'\}) \leq \omega_{\overline R}$ for all $u' \neq u \in S_M$. Note, the function $\rho_X(\cdot)$ in Definition \ref{defn:rdr} is replaced by $\rho'_X(\cdot)$ depending on the status of $\mathbb{I}_{v}^{OL}$. Here, the reward $\omega_{\overline R} = 2$ if one of two events happen: (i) there exists an $M$-active path in $X'$ with delayed-distance less than $dd(S_F,v) - \tau_v$ from some $u \in S_M$ to $v$ that does not overlap with any path from $S_F$ to $v$ or (ii) there exists an $M$-active path in $X'$ with delayed-distance less than $dd(S_F,v) + \tau_v$ from some $u \in S_M$ to $v$ that overlaps with some path from $S_F$ to $v$. On the other hand, $\omega_{\overline R} = 1$ if there exists an $M$-active path in $X'$ with delayed-distance in the range $[dd(S_F,v) - \tau_v, dd(S_F,v) + \tau_v]$ and there are no $M$-active paths with delayed-distance less than $dd(S_F,v) - \tau_v$ from some $u \in S_M$ to $v$ that does not overlap with any path from $S_F$ to $v$. It follows that $\varrho_1 = \varrho_2$.
\end{proof}

Lemmas \ref{lem:equiv_lower} and \ref{lem:equiv_upper} extend Lemma \ref{lem:borgs} by establishing a connection between the probability of a node receiving a particular reward value and RDR coverage weights. As remarked in \S~\ref{sec:prob_defn}, in place of step reward, we could use arbitrary non-increasing reward functions if more fine-grained control over the desirable behaviour is required. We can find the corresponding bounding functions needed for SA using the step functions corresponding to the Riemann sums \cite{hughes2020calculus} traditionally used in numerical integration. The tightness of the bounding functions improves as the number of subintervals increases. Notably, some care is required in ensuring the resulting bounding functions are submodular where consecutive step sizes must monotonically increase/decrease.

\section{Importance Sampling}
\label{sec:importance_sampling}
In this section, we describe how importance sampling can be used in our framework to reduce the sample complexity by analyzing the variance of the random variables associated with our unbiased estimator. Having established reward equivalence for both our upper and lower bounding functions, all of the analysis in this section applies to both functions. Therefore, to simplify the exposition, we describe the idea behind importance sampling and the resulting reverse sampling framework for a single abstracted objective $\sigma(\cdot)$ which could be \LL{instantiated as the upper or the lower bounding function, $\underline{\mu}(\cdot)$ or $\overline{\mu}(\cdot)$, of the mitigation objective.}  

\stitle{Unbiased Estimators.} Unlike the IM problem, where all nodes are candidates to be influenced, in the the MM problem, only those nodes that are influenced by the misinformation are candidates to contribute reward. \LL{Unlike targeted IM \cite{li2015real}, these nodes are not known a priori nor can they be precomputed.} As such, the uniform sampling approach leveraged by random RR sets for the IM problem is not directly useful for our setting. In theory, we can apply Rejection Sampling (RS) by selecting source nodes $v$ for our random RDR sets uniformly at random from $G$ and define the corresponding random RDR set as empty if $v \not\in R_F^X$. However, RS is best suited when the target probability is high and becomes less practical as the events become rarer.

A more sophisticated approach that yields improved sample efficiency for estimating rare events is Importance Sampling (IS). IS has been successfully leveraged for the \emph{targeted} IM problem \cite{li2015real}. IS estimates the expected value of a function $f$ in a probability space $P$ via sampling from another \emph{proposal distribution} $Q$, then re-weights the samples by an \emph{importance factor} for unbiased estimation. When applying IS, ideally $Q$ is chosen to support efficient sampling.

\stitle{IS for RDR Sets.} To apply IS in our setting we select the root for a random RDR set uniformly at random from the set $R_F^X$.  This ensures that the corresponding RDR set is non-empty. In other words, let $P$ be the probability space of RDR sets generated when the root node is selected uniformly at random from $G$. Then, we define $Q$ as a subspace of $P$ that corresponds to the space of \emph{only} non-empty samples of $P$. Specifically, those RDR sets for which the root $v \in R_F^X$. It remains to define an appropriate importance factor to ensure we have an unbiased estimator for $\sigma(S_M)$. We let $INF_F$ denote the expected number of activated \emph{non-seed} nodes due to seed set $S_F$ in the absence of any mitigating campaign. In other words, $INF_F$ is the expected influence of campaign $F$ in the absence of any mitigating campaign while ignoring the activation of seed nodes. Consider a set $S$ and a random RDR set $R_i(v)$ rooted at $v$ generated with IS as defined above. Define the following random variable:
\begin{equation}
Y_i(S) = \begin{cases}
\LL{S\cap R_i(v) =  \omega_{R_i(v)}} &\mbox{if } S \: \text{covers} \: R_i(v) \\
0 & \mbox{otherwise}
\end{cases}
\end{equation}
Then, we have the following lemma.
\begin{lem}
\label{lem:IS_equiv}
Given a random RDR set $R_i(v)$ generated with importance sampling rooted at $v$, for any set $S \subseteq V$, we have, $\sigma(S) = \mathbb{E}[Y_i(S)] \cdot INF_F$.
\end{lem}

\begin{proof}
Let $\eta_X(S_F \rightarrow v) = 1$ be an indicator for if $S_F$ can reach $v$ in possible world $X$. Let (1) $\eta_X(S \models v) = 1$ if $S$ can reach $v$ in $X$ along a path for which all tie-breaks are won \LL{by $M$} and achieves reward $1$ at $v$, (2) $\eta_X(S \models v) = 2$ if $S$ can reach $v$ in $X$ along a path for which all tie-breaks are won \LL{by $M$} and achieves reward $2$ at $v$ and (3) \LL{$\eta_X(S \models v) = 0$}  otherwise. When the root $v$ of $R_i$ is selected uniformly at random from $R_F^X$, we have
\begin{align*}
    \sigma(S) &= \sum_v \Big [ 2 \Pr_X ( \eta_X(S_F \rightarrow v) \; \land \; \Pr_X(\eta_X(S \models v) = 2) \Big ] \\
    & \qquad + \Big [ \Pr_X ( \eta_X(S_F \rightarrow v) \; \land \; \Pr_X(\eta_X(S \models v) = 1) \Big ] \\
    &= \sum_v \Pr_X ( \eta_X(S_F \rightarrow v) \cdot \Big [ 2 \Pr_X(\eta_X(S \models v) = 2) \\
    & \qquad \qquad + \Pr_X(\eta_X(S \models v) = 1) \Big ] \\
    &= INF_F \cdot [2 \Pr_{X,v \in R_F^X}(S \cap R_i(v) = 2) \\
    & \qquad \qquad + \Pr_{X,v \in R_F^X}(S \cap R_i(v) = 1)] \\
    &= INF_F \cdot \mathbb{E}_{X,v \in R_F^X}[Y_i(S)]
\end{align*}
Thus, the lemma is proved.
\end{proof}

Lemma \ref{lem:IS_equiv} states that we can estimate the expected reward of the mitigation campaign using random RDR sets generated with IS. Let $\mathcal{R}$ be a collection of $\theta$ random RDR sets generated with IS and let $\mathcal{W}_{\mathcal{R}}(S)$ be the total weight of RDR sets in $\mathcal{R}$ covered by a node set $S$. Then, based on Lemmas \ref{lem:equiv_lower}, \ref{lem:equiv_upper} and \ref{lem:IS_equiv}, we can prove:
\begin{cor}
\label{cor:equiv}
$\mathbb{E} \big [ \frac{\mathcal{W}_{\mathcal{R}}(S)}{\theta} \big ] \cdot INF_F = \sigma(S)$
\end{cor}

\stitle{Concentration Bounds.} Next, we analyze the random variables associated with random RDR sets generated using IS. In particular, we show they have smaller variances than random RDR sets generated by RS and, as a consequence, fewer samples are required by our reverse sampling framework. Define the random variable $Z_i(S) = \frac{Y_i(S) \cdot INF_F}{n}$. Notice that the means of $Y_i(S)$ and $Z_i(S)$ are $\mathbb{E}[Y_i(S)] = \frac{\sigma(S)}{INF_F}$ and $\mathbb{E}[Z_i(S)] = \mathbb{E}[Y_i(S)] \cdot \frac{INF_F}{n} = \frac{\sigma(S)}{n}$ respectively. If we construct a set of random variables $Z_1(S), \dots, Z_{\theta}(S)$, observe that $\frac{n}{\theta} \sum_{i = 1}^{\theta} Z_i(S)$ is an empirical estimate of $\sigma(S)$. An important challenge is that $INF_F$ is \#P-hard to compute. We overcome this challenge by computing an approximation of $INF_F$, denoted $\hat{INF_F}$, and define the random variable
\begin{equation}
\label{eqn:z_hat}
\hat{Z}_i(S) = \frac{X_i(S) \cdot \hat{INF_F}}{n}    
\end{equation}
where $\mathbb{E}[\hat{Z}_i(S)] = \frac{\sigma(S)}{n} \frac{\hat{INF_F}}{INF_F}$. Notice that estimating $INF_F$ is a standard \emph{influence estimation} task, and as such the random variables associated with estimating it form a \emph{martingale} \cite{tang2015influence}. Thus, we can leverage existing solutions for the influence estimation problem \cite{nguyen2017outward} to efficiently compute an $(\epsilon,\delta)$-approximation of $INF_F$. Suppose we have computed a value for $\INF$ with error $\epsilon'$ that holds with probability $1 - \delta'$. In order to ease the exposition, we introduce the following definition.

\begin{defn}
The \emph{misinformation sampling error ratio} is defined as $\Gamma = \frac{\INF}{INF_F}$ where $(1 - \epsilon') \leq \Gamma \leq (1 + \epsilon')$ holds with probability at least $1 - \delta'$.
\end{defn}
The variance of $\hat{Z}_i(S)$ satisfies the following inequality.
\begin{prop}
$\Var [\hat{Z}_i(S)] \leq 2 \Gamma \frac{\sigma(S)}{n} \frac{\INF}{n}$.
\end{prop}
\begin{proof}
Define the following random variables:
\begin{equation*}
    Y_i^{(1)} = 
    \begin{cases}
    1 &\mbox{if } S \: \text{covers} \: R_i \; \text{with weight 1} \\
    0 & \mbox{otherwise}
    \end{cases} 
\end{equation*}
\begin{equation*}
    Y_i^{(2)} = 
    \begin{cases}
    1 &\mbox{if } S \: \text{covers} \: R_i \; \text{with weight 2} \\
    0 & \mbox{otherwise}
    \end{cases} 
\end{equation*}
Then, we can re-write the random variable $\hat{Z}_i(S)$ as $\hat{Z}_i(S) = \frac{\hat{INF}_F}{n} \cdot (Y_i^{(1)} + 2 Y_i^{(2)})$. As a result, we have
\begin{equation*}
    \mathbb{E}[\hat{Z}_i(S)] = \frac{\hat{INF}_F}{n} \cdot (\mathbb{E}[Y_i^{(1)}] + 2 \mathbb{E}[Y_i^{(2)}])
\end{equation*}
Now, we can bound the variance of $\hat{Z}_i(S)$.
\begin{align*}
\Var [\hat{Z}&_i(S)] = \Var \Big [ \frac{\hat{INF_F}}{n} \cdot (Y_i^{(1)} + 2 Y_i^{(2)}) \Big ] \\
    &= \frac{\hat{INF_F}^2}{n^2} \Big ( \Var [Y_i^{(1)}] + 4 \Var [Y_i^{(2)}] + 4 \Cov [Y_i^{(1)},Y_i^{(2)}] \Big ) \\
    &\leq \frac{\hat{INF_F}^2}{n^2} \Big ( \mathbb{E}[Y_i^{(1)}] + 4 \mathbb{E}[Y_i^{(2)}] \Big )
    \leq 2 \Gamma \cdot \frac{\sigma(S)}{n} \frac{\hat{INF}_F}{n}
\end{align*}
Notice that the mutual exclusion of the events described by $Y_i^{(1)}$ and $Y_i^{(2)}$ ensures that $\Cov[Y_i^{(1)},Y_i^{(2)}] \leq 0$.
\end{proof}
Now, we seek a form of Chernoff bounds for our newly defined random variables $\hat{Z}_i(S)$ in order to be able to analyze the performance guarantees of our reverse sampling framework. Importantly, we must account for the error associated with the estimation of $INF_F$. This hurdle was not encountered in previous applications of importance sampling to targeted IM \cite{li2015real} since the target set of nodes is \emph{pre-defined}. We make use of martingale-based concentration bounds in the following.
\begin{defn}[Martingale]
A sequence of random variables $Y_1, Y_2, Y_3, \dots$ is a \emph{martingale} if and only if $\mathbb{E}[|Y_i|] < + \infty$ and $\mathbb{E}[Y_i | Y_1, Y_2, \dots, Y_{i-1}] = Y_{i-1}$ for any $i$.
\end{defn}
It is straightforward to show that the random variables $\hat{Z}_i(S)$ form a martingale. Thus, the Chernoff bounds for martingales given in Lemma 2 of \cite{tang2015influence} let us derive the following concentration bounds for the random variables $\hat{Z}_i(S)$ associated with our RDR sets generated with IS by plugging in the variance derived above. Note, we assume that $\hat{INF}_F \leq \frac{n}{2}$ as a necessary boundary condition for the derivation of our concentration bounds. Further, in our experiments the condition always held since an unrealistic number of misinformation seeds would be required to influence over half the networks considered.
\begin{lem}
\label{lem:new_concentration_bounds}
Given a fixed collection of $\theta$ RDR sets $\mathcal{R}$ constructed with importance sampling and a seed set $S$, let $\Lambda(S) = \frac{\hat{INF}_F}{n} \mathcal{W}_{\mathcal{R}}(S)$ be the normalized weighted coverage of $S$ in $\mathcal{R}$. For any $\lambda > 0$ we have,
\begin{equation}
    \Pr \Big [ \Lambda(S) - \sigma(S) \cdot \Gamma \frac{\theta}{n} \geq \lambda \Big ] \leq \exp \Big ( \frac{- \lambda^2}{\frac{2}{3} \lambda + 4 \sigma(S) \Gamma \frac{\theta}{n} \frac{\hat{INF}_F}{n}} \Big )
\end{equation}
\begin{equation}
    \Pr \Big [ \Lambda(S) - \sigma(S) \cdot \Gamma \frac{\theta}{n} \leq - \lambda \Big ] \leq \exp \Big ( \frac{- \lambda^2}{4 \sigma(S) \Gamma \frac{\theta}{n} \frac{\hat{INF}_F}{n}} \Big )
\end{equation}
\end{lem}
\begin{proof}
We begin by recalling the concentration bounds for martingales. Let $M_1, M_2, \dots$ be a martingale such that $|M_1| \leq a$, $|M_j - M_{j-1}| \leq a \; \forall j \in [2,\theta]$ and
\begin{equation*}
    \Var[M_1] + \sum_{j=2}^{\theta} \Var[M_j | M_1, M_2, \dots, M_{j-1}] \leq b
\end{equation*}
Where $\Var[\cdot]$ denotes the variance of a random variable. Then, for any $\lambda > 0$,
\begin{equation}
\label{eqn:cb_bound_proof_1}
    \Pr[M_{\theta} - \mathbb{E}[M_{\theta}] \geq \lambda] \leq \exp \Big ( \frac{- \lambda^2}{\frac{2}{3} a \lambda + 2 b} \Big )
\end{equation}
and,
\begin{equation}
\label{eqn:cb_bound_proof_2}
    \Pr[M_{\theta} - \mathbb{E}[M_{\theta}] \leq - \lambda] \leq \exp \Big ( \frac{ - \lambda^2}{2 b} \Big )
\end{equation}
Next, we apply the above concentration bounds on the martingale formed by the random variable $\hat{Z}_i(S)$ associated with our RDR sets generated with importance sampling. First, since each RDR set $R_i$ is generated randomly and independently of all the prior RDR sets, we have
\begin{equation*}
    \mathbb{E}[\hat{Z}_i(S) | \hat{Z}_1(S), \hat{Z}_2(S), \dots, \hat{Z}_{i-1}(S)] = \mathbb{E}[\hat{Z}_i(S)] = \Gamma \frac{\sigma(S)}{n}.
\end{equation*}
Let $p = \Gamma \frac{\sigma(S)}{n}$ and $M_i = \sum_{j=1}^{i} (\hat{Z}_j(S) - p)$. Then, we have $\mathbb{E}[M_i] = 0$ and
\begin{equation*}
    \mathbb{E}[M_i | M_1, M_2, \dots, M_{i-1}] = M_{i-1}
\end{equation*}
Therefore, $M_1, M_2, \dots, M_{\theta}$ is a martingale. We have $|M_1| \leq \frac{2 \INF}{n} \leq 1$ assuming $\INF \leq \frac{n}{2}$ and $|M_j - M_{j-1}| \leq 1 \; \forall j \in [2,\theta]$. We also have
\begin{multline}
    \Var[M_1] + \sum_{j=2}^{\theta} \Var[M_j | M_1, M_2, \dots, M_{j-1}] \\
    \qquad = \sum_{j=1}^{\theta} \Var[\hat{Z}_j(S)] \leq 2 \theta \Gamma \cdot \frac{\sigma(S)}{n} \frac{\INF}{n}
\end{multline}
Plugging into \ref{eqn:cb_bound_proof_1} and \ref{eqn:cb_bound_proof_2} completes the proof.
\end{proof}

We will make use of the above concentration bounds to derive the approximation guarantees of our sampling framework and to establish appropriate parameter settings.
\section{Reverse Sampling Framework}
\label{sec:reverse_sampling}

Recently, Tang et al. \cite{tang2018online} introduced the \emph{OPIM} approach to the online version of the IM problem. Interestingly, an adaptation of \emph{OPIM} to the traditional IM problem yields state-of-the-art performance. Unlike \emph{IMM} \cite{tang2015influence}, which uses the same collection $\mathcal{R}$ of RR sets for constructing the solution seed set $S^*$ and deriving its approximation guarantees, \emph{OPIM} generates a solution on one collection of RR sets $\mathcal{R}_1$ and then derives its approximation guarantees using $\mathcal{R}_1$ and an independent collection of RR sets $\mathcal{R}_2$. In particular, the concentration bounds leveraged by \emph{OPIM} require that $S^*$ be a fixed seed set independent of the RR sets on which it is being evaluated. Intuitively, we can think of $\mathcal{R}_1$ as a set of \emph{nominators} that nominate $S^*$ as the IM solution and $\mathcal{R}_2$ as the set of \emph{assessors} that determine whether $S^*$ is a good enough solution. Notice that, if $S^*$ is not independent of $\mathcal{R}_2$, then the evaluation of $S^*$ could be biased.

We make the following important observation that motivates our framework: the combination of nominators and assessors leveraged by the \emph{OPIM} algorithm does not depend on any particular properties of the IC propagation model, as long as it satisfies activation equivalence (Lemma \ref{lem:borgs}). Thus, the \emph{OPIM} framework can be extended as follows. Specifically, by establishing weighted versions of Lemma \ref{lem:borgs} for our RDR sets in the \model model, we can employ a nominator-assessor framework to derive approximation guarantees for our upper and lower bounding mitigation functions. Thus, our algorithm and approximation guarantees apply to any propagation models satisfying reward equivalence (Lemmas \ref{lem:equiv_lower} and \ref{lem:equiv_upper}).

\stitle{MM Solution.} Our framework for finding solutions to the upper and lower bounding objectives to the MM problem, \emph{NAMM (Nominators and Assessors for Misinformation Mitigation)} is presented in Algorithm \ref{alg:namm}. During its execution, \emph{NAMM} invokes the standard greedy algorithm for weighted maximum coverage to obtain a size-$k$ seed set $S^{*}$. All of the analysis in this section applies to both of the upper or lower bounding functions, $\underline{\mu}(\cdot)$ or $\overline{\mu}(\cdot)$, of the mitigation objective. Therefore, to simplify the exposition, we refer to a single abstracted objective $\sigma(\cdot)$ in our algorithm description.
\begin{algorithm}
\caption{NAMM}\label{alg:namm}
\algorithmicrequire $G$, $\epsilon \geq 0$, $0 < \delta < 1$, $k$ \\
\algorithmicensure An $(1-1/e-\epsilon)$-optimal solution $S^{*}$
\begin{algorithmic}[1]
\State $\delta' \gets \frac{\delta}{9}$, $\epsilon' \gets \frac{\epsilon}{2}$, $\Delta \gets \delta - \delta'$
\State compute $\INF$; an $(\epsilon',\delta')$-approximation of $INF_F$
\State set $N_{max}$ according to (\ref{eqn:n_max})
\State $N_0 = N_{max} \cdot \epsilon^2 \frac{LB}{n}$;
\State generate two collections, $\mathcal{R}_1$ and $\mathcal{R}_2$, of random RDR sets where $|\mathcal{R}_1| = |\mathcal{R}_2| = N_0$;
\State $i_{max} = \lceil \log_2 (\frac{N_{max}}{N_0}) \rceil$;
\For{$i \gets 1$ \textbf{to} $i_{max}$}
\State $S^{*} \gets \texttt{WeightedMaxCover}(\mathcal{R}_1,k,n)$
\State compute $\sigma^{l}(S^{*})$ and $\sigma^{u}(S^{o})$ by (\ref{eqn:lower_bound}) and (\ref{eqn:upper_bound}) respectively, setting $\delta_1 = \delta_2 = \Delta / (3 i_{max})$
\State $\alpha \gets \sigma^{l}(S^{*}) / \sigma^{u}(S^{o})$
\If{$\alpha \geq (1 - 1/e - \epsilon)$ \textbf{or} $i = i_{max}$}
\State \textbf{return} $S^{*}$
\EndIf
\State double the sizes of $\mathcal{R}_1$ and $\mathcal{R}_2$ with new random RDR sets
\EndFor
\end{algorithmic}
\end{algorithm}
The approximation guarantee of \emph{NAMM} relies on two critical lemmas that establish an upper bound on the mitigation of the optimal solution $S^{o}$ ($\sigma^{u}(S^{o})$) and a lower bound on the mitigation of the current solution ($\sigma^{l}(S^{*})$). After establishing an upper bound, $N_{max}$, on the number of RDR sets required in the worst-case, we show that \emph{NAMM} achieves a $(1 - 1/e - \epsilon)$-approximation by leveraging our newly defined concentration bounds (Lemma \ref{lem:new_concentration_bounds}). Importantly, the error associated with estimating $INF_F$ must be carefully accounted for in deriving upper and lower bounds to ensure the desired approximation guarantees.

\underline{\etitle{Deriving $N_{max}$.}} Tang et al.\ \cite{tang2015influence} derived a threshold for the maximum number of RR sets required by their IMM algorithm to ensure that a $(1 - 1/e - \epsilon)$ approximation guarantee holds with probability at least $1- \delta$. We derive the corresponding threshold for the MM problem. Notably, the derivation requires a careful accounting of the error associated with estimating $INF_F$.

Let $\mathcal{R}$ be a collection of $\theta$ random RDR sets generated with IS, $\mathcal{W}_{\mathcal{R}}(S)$ be the total weight of RDR sets in $\mathcal{R}$ covered by a node set $S$, and $F_{\mathcal{R}}(S) = \frac{1}{\theta} \sum_i \hat{Z}_i(S)$. Now, consider the size-$k$ node set $S^o$ with the maximum expected mitigation. Let $OPT = \sigma(S^o)$. We have $n \cdot F_{\mathcal{R}}(S^o)$ is an unbiased estimator of $\Gamma \cdot OPT$. As such, we can use our newly derived Chernoff bounds to show the following.
\begin{lem}
\label{lem:theta_one}
Let $\delta_1 \in (0,1)$, $\epsilon_1 > 0$ and
\begin{equation}
    \theta_1 = \frac{4 \INF \log(\frac{1}{\delta_1})}{\epsilon_1^2 (1 - \epsilon') OPT}.
\end{equation}
If $\theta \geq \theta_1$, then $n \cdot F_{\mathcal{R}}(S^o) \geq (1 - \epsilon')(1-\epsilon_1) \cdot OPT$ holds with at least $1 - \delta_1$ probability.
\end{lem}
\begin{proof}
Let $p = \mathbb{E}[F_{\mathcal{R}}(S^o)]$. By Lemma \ref{lem:IS_equiv} and Equation \ref{eqn:z_hat},
\begin{equation*}
    p = \mathbb{E}[F_{\mathcal{R}}(S^o)] = \Gamma \cdot \frac{\sigma(S^o)}{n} = \Gamma \cdot \frac{OPT}{n}.
\end{equation*}
By Lemma \ref{lem:new_concentration_bounds},
\begin{align*}
    \Pr[n & \cdot F_{\mathcal{R}}(S^o) \leq (1-\epsilon')(1-\epsilon_1) \cdot OPT] \\
        &= \Pr[n \cdot F_{\mathcal{R}}(S^o) \leq (1-\epsilon')(1-\epsilon_1) \cdot \frac{n p}{\Gamma}] \\
        &= \Pr[\theta \cdot F_{\mathcal{R}}(S^o) \leq (1-\epsilon')(1-\epsilon_1) \cdot \frac{\theta p}{\Gamma}] \\
        &\leq \Pr[\theta \cdot F_{\mathcal{R}}(S^o) \leq \Gamma (1-\epsilon_1) \cdot \frac{\theta p}{\Gamma}] \\
        &= \Pr \Big [ \sum_i \hat{Z}_i(S^o) - \theta p \leq -\epsilon_1 \theta p \Big ] \\
        &= \Pr \Big [ \sum_i \hat{Z}_i(S^0) - \sigma(S^o) \Gamma \frac{\theta}{n}  \leq -\epsilon_1 \sigma(S^o) \Gamma \frac{\theta}{n} \Big ] \\
        &\leq \exp \Big ( \frac{- \epsilon_1^2 \sigma(S^o)^2 \Gamma^2 (\frac{\theta}{n})^2}{4 \sigma(S^o) \Gamma \frac{\theta}{n} \frac{\INF}{n}} \Big ) \\
        &= \exp \Big ( \frac{- \epsilon_1^2 \sigma(S^o) \Gamma \theta}{4 \INF} \Big ) \leq \exp \Big ( \frac{- \epsilon_1^2 \sigma(S^o) (1 - \epsilon') \theta}{4 \INF} \Big ) \\
        &\leq \delta_1.
\end{align*}
Thus, the lemma is proved.
\end{proof}
Suppose that $n \cdot F_{\mathcal{R}}(S^o) \geq (1 - \epsilon')(1-\epsilon_1) \cdot OPT$ holds, by the properties of the greedy approach,
\begin{align}
    n \cdot F_{\mathcal{R}}(S^*) &\geq \Big ( 1-\frac{1}{e} \Big ) \cdot n \cdot F_{\mathcal{R}}(S^o) \nonumber \\
    &\geq \Big ( 1-\frac{1}{e} \Big )(1 - \epsilon')(1-\epsilon_1) \cdot OPT. \label{eqn:seed_mit_bound}
\end{align}
Intuitively, this indicates that the expected mitigation of $S^*$ is likely to be large, since $n \cdot F_{\mathcal{R}}(S^*)$ is an indicator of $\Gamma \sigma(S^*)$. This is formalized in the following lemma.
\begin{lem}
\label{lem:theta_two}
Let $\delta_2 \in (0,1)$, $\epsilon_1 < \epsilon$ and
\begin{equation}
    \theta_2 = \frac{(6 + 4 \epsilon') (1 - 1/e) \cdot n \log(\frac{\binom{n}{k}}{\delta_2})}{3 OPT [\epsilon - (1 - 1/e)(1 - (1 - \epsilon')(1 - \epsilon_1))]^2}.
\end{equation}
If Equation \ref{eqn:seed_mit_bound} holds and $\theta \geq \theta_2$, then with at least $1 - \delta_2$ probability, $\sigma(S^*) \geq (1 - 1/e - \epsilon) \cdot OPT$.
\end{lem}
\begin{proof}
Let $S$ be an arbitrary size-$k$ seed set. We say $S$ is \emph{bad} if $\sigma(S) < (1 - 1/e - \epsilon) \cdot OPT$. To prove the lemma, we show that each bad size-$k$ seed set has at most $\delta_2 / \binom{n}{k}$ probability to be returned by applying the greedy algorithm to a collection of $\theta_2$ RDR sets. This suffices to establish the lemma because (1) there exist only $\binom{n}{k}$ bad size-$k$ seed sets, and (2) if each of them has at most $\delta_2 / \binom{n}{k}$ probability to be returned, then by the union bound, there is at least $1 - \delta_2$ probability that none of them is output by the greedy algorithm.

Consider any bad size-$k$ seed set $S$. Let $p = \mathbb{E}[F_{\mathcal{R}}(S)] = \Gamma \frac{\sigma(S)}{n}$. We have
\begin{align}
\Pr[n & \cdot F_{\mathcal{R}}(S) - \Gamma \sigma(S) \geq \epsilon_2 \cdot OPT] \nonumber \\
    &= \Pr[n \cdot F_{\mathcal{R}}(S) \geq \Gamma \sigma(S) + \epsilon_2 \cdot OPT] \nonumber \\
    &\geq \Pr[n \cdot F_{\mathcal{R}}(S) \geq (1+\epsilon') (1 - 1/e - \epsilon) OPT + \epsilon_2 \cdot OPT] \nonumber \\
    &= \Pr[n \cdot F_{\mathcal{R}}(S) \geq ((1+\epsilon') (1 - 1/e - \epsilon) + \epsilon_2) \cdot OPT] \label{eqn:bad_seed_bound}
\end{align}
We set $\epsilon_2$ such that the multiplicative factor of $OPT$ in Equation \ref{eqn:bad_seed_bound} is equal to the one in Equation \ref{eqn:seed_mit_bound}.
\begin{equation}
\label{eqn:eps_two}
\epsilon_2 = (1-\epsilon')(1-\epsilon_1)(1 - 1/e) - (1+\epsilon')(1 - 1/e - \epsilon).
\end{equation}
Then, we can apply our Chernoff bounds by re-writing $\Pr[n \cdot F_{\mathcal{R}}(S) - \Gamma \sigma(S) \geq \epsilon_2 \cdot OPT] = \Pr[\theta \cdot F_{\mathcal{R}}(S) - \Gamma \sigma(S) \frac{\theta}{n} \geq \epsilon_2 \cdot \frac{\theta}{n} \cdot OPT]$ and letting $\lambda = \epsilon_2 \cdot \frac{\theta}{n} \cdot OPT$ to get
\begin{align*}
    \text{r.h.s. of Eq. \ref{eqn:bad_seed_bound}} &\leq \exp \Bigg ( \frac{- \lambda^2}{\frac{2}{3} \lambda + 4 \sigma(S) \Gamma \frac{\theta}{n} \frac{\INF}{n}} \Bigg ) \\
    &= \exp \Bigg ( \frac{- \epsilon_2^2 \cdot \theta^2 \cdot OPT^2}{n^2 (\frac{2}{3} \frac{\epsilon_2 \cdot \theta \cdot OPT}{n} + 4 \sigma(S) \Gamma \frac{\theta}{n} \frac{\INF}{n})} \Bigg ) \\
    &\leq \exp \Bigg ( \frac{- \epsilon_2^2 \cdot \theta^2 \cdot OPT^2}{\frac{2}{3} \epsilon_2 n \theta \cdot OPT + 4 \theta (1 - 1/e -\epsilon) OPT \cdot \Gamma \cdot \INF} \Bigg ) \\
    &\leq \exp \Bigg ( \frac{- \epsilon_2^2 \theta OPT}{\frac{2}{3} \epsilon_2 n + 4 (1 - 1/e -\epsilon) (1+\epsilon') \INF} \Bigg ) \\
    &\leq \exp \Bigg ( \frac{- \epsilon_2^2 \theta OPT}{\frac{2n}{3} (1 - 1/e)(3+2\epsilon')} \Bigg ) \\
    &\leq \frac{\delta_2}{\binom{n}{k}}.
\end{align*}
Therefore, the lemma holds.
\end{proof}
Lemmas \ref{lem:theta_one} and \ref{lem:theta_two} lead to the following theorem.
\begin{thm}
\label{thm:n_max}
Given any $\epsilon_1$, $\epsilon'$ such that $\epsilon_1 \leq \epsilon$ and any $\delta_1, \delta_2 \in (0,1)$ with $\delta_1 + \delta_2 + \delta' \leq \delta$, setting $\theta \geq \argmax\{ \theta_1, \theta_2 \}$ ensures that the node selection phase of \emph{NAMM} returns a $(1 - 1/e - \epsilon)$-approximate solution with at least $1 - \delta$ probability.
\end{thm}
\begin{proof}
By Lemma \ref{lem:theta_two}, $\sigma(S^*) \geq (1 - 1/e - \epsilon) \cdot OPT$ holds with at least $1 - \delta_2$ probability under the condition that Equation \ref{eqn:seed_mit_bound} holds. And by Lemma \ref{lem:theta_one}, Equation \ref{eqn:seed_mit_bound} holds with at least $1 - \delta_1$ probability. By the union bound, $\sigma(S^*) \geq (1 - 1/e - \epsilon) \cdot OPT$ holds with at least $1 - \delta_1 - \delta_2 \geq 1 - \delta$ probability. Thus, the theorem is proved.
\end{proof}
Now, a natural question is, how should we select $\epsilon'$, $\epsilon_1$, $\delta_1$, and $\delta_2$ that adhere to the conditions of Theorem \ref{thm:n_max} in order to minimize $\theta$? Assume that $OPT$ is known. We are trying to minimize $\theta^* = \argmax \{ \theta_1, \theta_2 \}$ subject to $\delta_1 + \delta_2 + \delta' \leq \delta$. Following an approach similar to that of \cite{tang2015influence}, we set $\delta_1 = \delta_1 = \frac{4}{9} \delta$ and $\delta' = \frac{1}{9} \delta$ and set $\theta_1 = \theta_2$ to derive an approximately minimal value for $\theta^*$. Solving for $\epsilon_1$ and plugging the result back into $\theta_1$ yields
\begin{equation}
\theta^* \leq \frac{8 n (3 + \epsilon') (1 - 1/e) [\ln \frac{9}{4 \delta} + \ln \binom{n}{k}]}{3 \cdot OPT [\epsilon(1 + \epsilon') - 2 \epsilon' (1 - 1/e)]^2}.
\end{equation}
Thus, we have the following result.
\begin{lem}
\label{lem:n_max}
Let $\mathcal{R}$ be a collection of random RDR sets and $S^*$ be a size-$k$ seed set generated by applying the greedy algorithm on $\mathcal{R}$. For fixed $\epsilon$, $\epsilon'$, and $\delta$, if $\delta' \leq \frac{\delta}{9}$ and
\begin{equation*}
    |\mathcal{R}| \geq \frac{8 n (3 + \epsilon') (1 - 1/e) [\ln \frac{9}{4 \delta} + \ln \binom{n}{k}]}{3 \cdot OPT [\epsilon(1 + \epsilon') - 2 \epsilon' (1 - 1/e)]^2}
\end{equation*}
then $S^*$ is a $(1-1/e-\epsilon)$-approximate solution with at least $1-\delta$ probability.
\end{lem}
Define $\Delta = \delta - \delta'$. Based on Lemma \ref{lem:n_max}, we define a value $N_{max}$:
\begin{equation}
\label{eqn:n_max}
    N_{max} = \frac{8 n (3 + \epsilon') (1 - 1/e) [\ln \frac{27}{4 \Delta} + \ln \binom{n}{k}]}{3 \cdot LB [\epsilon(1 + \epsilon') - 2 \epsilon' (1 - 1/e)]^2},
\end{equation}
which is an upper bound on the number of RDR sets needed to guarantee a $(1-1/e-\epsilon)$ approximation with at least $1-\Delta/3$ probability when $LB \leq OPT$ is a lower bound of the optimal mitigation.

\underline{\etitle{Deriving $LB$.}} In \emph{OPIM}, a crude lower bound of $k$ is used that ensures the number of iterations is bounded by $O(\log n)$. By contrast, due to the objective of the MM problem, the same lower bound is no longer valid. Specifically, each seed selected in the IM problem ensures that, independent of the graph structure, the seed will contribute one activation to the objective. \LL{However}, the objective of the MM problem requires that a node $v$ \textit{would} have been reached by $F$ \LL{in order for} $M$ to achieve any reward at $v$. \textit{As such, the ability for seeds in $S_M$ to provide reward is intimately connected to the set $S_F$ and the graph structure.} A pessimistic lower bound can be derived by considering the probability that a randomly selected node is reached by $F$.

To tackle this, we adopt ideas from the classical \emph{Maximum Influence Arborescence (MIA)} \cite{wang2012scalable} approach to the IM problem to derive a lower bound on $OPT$. To recall, the MIA framework assumes that influence only travels via the paths of maximum influence in the network and leverages the resulting structures to estimate influence spread. A \emph{MIA} structure constructed from a seed set $S$ computes an activation probability $ap(v)$ for each node $v$ in the arborescence. Notably, the assumption that influence only spreads through the maximum influence paths implies that the resulting activation probabilities lower bound the true activation probabilities. To see this, notice that influence spreading through paths \emph{other} than the maximum influence path can only increase the probability that any node becomes activated. We are interested in a lower bound $LB$ on the mitigation of $k$ seeds which depends critically on the behaviour of the spread of $S_F$. Our idea is to use a \emph{MIA} to lower bound the influence of $S_F$. Notice that the $ap(v)$ value gives a direct avenue to estimating the expected reward achieved at $v$ by selecting $v$ as a seed for the mitigating campaign. In particular, we select the top-$k$ nodes at depth $1$ in a \emph{MIA} constructed for $S_F$, ranked by activation probabilities. Thus, the expected mitigation of the selected nodes $S_{LB}$ is at least $LB \geq \sum_{v \in S_{LB}} ap(v)$. We replace $OPT$ in Equation \ref{eqn:n_max} with $LB$ and note that this setting ensures that the number of iterations is still bounded by $O(\log n)$.

\underline{\etitle{Deriving $\sigma^{l}(S^{*})$.}} Let $\Lambda_2(S^*) = \frac{\INF}{n} \mathcal{W}_{\mathcal{R}_2}(S^*)$ be the weighted coverage of $S^*$ in $\mathcal{R}_2$ and $\theta_2 = |\mathcal{R}_2|$. We have the following result.
\begin{lem}
\label{lem:lower_bound}
For any $\delta \in (0,1)$, we have
\begin{multline}
    \Pr \bigg [ \sigma(S^*) \geq \Bigg ( \bigg ( \sqrt{ \Lambda_2(S^*) + \frac{25 a}{36}} - \sqrt{a} \bigg )^2 - \frac{a}{36} \Bigg ) \frac{n}{\theta_2 (1+\epsilon')} \bigg ] \\
        \qquad \qquad \geq 1 - \delta
\end{multline}
where $a = \ln(1/\delta)$.
\end{lem}
\begin{proof}
Let $b = \sigma(S^*) \cdot \Gamma \cdot \frac{\theta_2}{n}$. We have
\begin{align*}
    \Pr \bigg [ & \sigma(S^*) < \Bigg ( \bigg ( \sqrt{ \Lambda_2(S^*) + \frac{25 a}{36}} - \sqrt{a} \bigg )^2 - \frac{a}{36} \Bigg ) \frac{n}{\theta_2 (1+\epsilon')} \bigg ] \\
    &\leq \Pr \bigg [ \sigma(S^*) < \Bigg ( \bigg ( \sqrt{ \Lambda_2(S^*) + \frac{25 a}{36}} - \sqrt{a} \bigg )^2 - \frac{a}{36} \Bigg ) \frac{n}{\theta_2 \cdot \Gamma} \bigg ] \\
    &= \Pr \bigg [ \sigma(S^*) \cdot \frac{\Gamma \cdot \theta_2}{n} < \bigg ( \sqrt{ \Lambda_2(S^*) + \frac{25 a}{36}} - \sqrt{a} \bigg )^2 - \frac{a}{36} \bigg ] \\
    &= \Pr \bigg [ b + \frac{a}{36} < \bigg ( \sqrt{ \Lambda_2(S^*) + \frac{25 a}{36}} - \sqrt{a} \bigg )^2 \bigg ] \\
    &= \Pr \bigg [ \sqrt{b + \frac{a}{36}} < \sqrt{ \Lambda_2(S^*) + \frac{25 a}{36}} - \sqrt{a} \bigg ] \\
    & \qquad + \Pr \bigg [ \sqrt{b + \frac{a}{36}} < \sqrt{a} - \sqrt{ \Lambda_2(S^*) + \frac{25 a}{36}} \bigg ] \\
    &= \Pr \bigg [ \sqrt{b + \frac{a}{36}} < \sqrt{ \Lambda_2(S^*) + \frac{25 a}{36}} - \sqrt{a} \bigg ] \\
    &= \Pr \bigg [ \sqrt{b + \frac{a}{36}} + \sqrt{a} < \sqrt{ \Lambda_2(S^*) + \frac{25 a}{36}} \bigg ] \\
    &= \Pr \bigg [ b + \sqrt{4ab + \frac{a^2}{9}} + \frac{a}{3} < \Lambda_2(S^*) \bigg ] \\
    &= \Pr \bigg [ \sqrt{4ab + \frac{a^2}{9}} + \frac{a}{3} < \Lambda_2(S^*) - b \bigg ] \\
    &= \Pr \bigg [ \sqrt{4ab + \frac{a^2}{9}} + \frac{a}{3} < \Lambda_2(S^*) - \sigma(S^*) \cdot \Gamma \cdot \frac{\theta_2}{n} \bigg ] \\
    &\leq \Pr \bigg [ \sqrt{4ab + \frac{a^2}{9}} + \frac{a}{3} < \Lambda_2(S^*) - \sigma(S^*) \cdot \Gamma \cdot \frac{\theta_2}{n} \bigg ] \\
    &\leq \exp \Bigg ( \frac{- \Big ( \sqrt{4ab + \frac{a^2}{9}} + \frac{a}{3} \Big )^2}{4 \sigma(S^*) \cdot \Gamma \cdot \frac{\theta_2}{n} \frac{\INF}{n} + \frac{2}{3} \Big ( \sqrt{4ab + \frac{a^2}{9}} + \frac{a}{3} \Big )} \Bigg ) \\
    &\leq \exp \Bigg ( \frac{- \Big ( \sqrt{4ab + \frac{a^2}{9}} + \frac{a}{3} \Big )^2}{4 \sigma(S^*) \cdot \Gamma \cdot \frac{\theta_2}{n} + \frac{2}{3} \Big ( \sqrt{4ab + \frac{a^2}{9}} + \frac{a}{3} \Big )} \Bigg ) \\
    &= \exp(-a) = \delta,
\end{align*}
where the second probability term vanishes due to
\begin{equation*}
    \sqrt{b + \frac{a}{36}} \geq \sqrt{\frac{a}{36}} = \sqrt{a} - \sqrt{\frac{25 a}{36}} \geq \sqrt{a} - \sqrt{\Lambda_2(S^*) + \frac{25 a}{36}}.
\end{equation*}
\end{proof}
Based on Lemma \ref{lem:lower_bound} and a parameter $\delta_2$ to be discussed shortly, we set
\begin{multline}
    \label{eqn:lower_bound}
    \sigma^{l}(S^{*}) = \Bigg ( \bigg ( \sqrt{ \Lambda_2(S^*) + \frac{25 \ln(1 / \delta_2)}{36}} - \sqrt{\ln(1 / \delta_2)} \bigg )^2 \\ - \frac{\ln (1 / \delta_2)}{36} \Bigg ) \frac{n}{\theta_2 (1+\epsilon')}
\end{multline}

\underline{\etitle{Deriving $\sigma^{u}(S^{o})$.}} We establish an upper bound of $\sigma(S^o)$ based on the weighted coverage of $S^*$ in $\mathcal{R}_1$, denoted as $\Lambda_1(S^*) = \frac{\INF}{n} \mathcal{W}_{\mathcal{R}_1}(S^*)$, by leveraging the property of the greedy algorithm that ensures $\Lambda_1(S^*) \geq (1 - 1/e) \Lambda_1(S^o)$.
\begin{lem}
\label{lem:upper_bound}
Let $\theta_1 = |\mathcal{R}_1|$. For any $\delta \in (0,1)$, we have
\begin{equation}
    \Pr \bigg [ \sigma(S^o) \leq \Bigg ( \sqrt{ \frac{\Lambda_1(S^*)}{1 - 1/e} + a} \; + \; \sqrt{a} \Bigg )^2 \frac{n}{\theta_1 (1 - \epsilon')} \bigg ] \geq 1 - \delta
\end{equation}
where $a = \ln(1/\delta)$.
\end{lem}
\begin{proof}
Let $b = \sigma(S^o) \cdot \Gamma \cdot \frac{\theta_1}{n}$. We have
\begin{align*}
    \Pr \bigg [ & \sigma(S^o) > \Bigg ( \sqrt{ \frac{\Lambda_1(S^*)}{1 - 1/e} + a} + \sqrt{a} \Bigg )^2 \frac{n}{\theta_1 (1-\epsilon')} \bigg ] \\
    &\leq \Pr \bigg [ \sigma(S^o) > \Bigg ( \sqrt{ \frac{\Lambda_1(S^*)}{1 - 1/e} + a} + \sqrt{a} \Bigg )^2 \frac{n}{\theta_1 \Gamma} \bigg ] \\
    &\leq \Pr \bigg [ \sigma(S^o) > \Bigg ( \sqrt{ \Lambda_1(S^o) + a} + \sqrt{a} \Bigg )^2 \frac{n}{\theta_1 \Gamma} \bigg ] \\
    &= \Pr \bigg [ \sigma(S^o) \cdot \Gamma \cdot \frac{\theta_1}{n} > \bigg ( \sqrt{ \Lambda_1(S^o) + a} + \sqrt{a} \bigg )^2 \bigg ] \\
    &= \Pr \Big [ b > ( \sqrt{ \Lambda_1(S^o) + a} + \sqrt{a} )^2 \Big ] \\
    &= \Pr \Big [ \sqrt{b} > \sqrt{ \Lambda_1(S^o) + a} + \sqrt{a} \Big ] \\
    &= \Pr \Big [ \sqrt{b} - \sqrt{a} > \sqrt{ \Lambda_1(S^o) + a} \Big ] \\
    &= \Pr \Big [ b + a - 2 \sqrt{ab} > \Lambda_1(S^o) + a \Big ] \\
    &= \Pr \Big [ - \sqrt{4ab} > \Lambda_1(S^o) - b \Big ] \\
    &= \Pr \bigg [ - \sqrt{4ab} > \Lambda_1(S^o) - \sigma(S^o) \cdot \Gamma \cdot \frac{\theta_1}{n} \bigg ] \\
    &\leq \Pr \bigg [ - \sqrt{4ab} > \Lambda_1(S^o) - \sigma(S^o) \cdot \Gamma \cdot \frac{\theta_1}{n} \bigg ] \\
    &\leq \exp \Bigg ( \frac{-4a \sigma(S^o) \cdot \Gamma \cdot \frac{\theta_1}{n}}{4 \sigma(S^o) \cdot \Gamma \cdot \frac{\theta_1}{n} \frac{\INF}{n}} \Bigg ) \\
    &= \exp \Bigg ( \frac{-a}{\frac{\INF}{n}} \Bigg ) \leq \exp(-a) = \delta.
\end{align*}
\end{proof}
Based on Lemma \ref{lem:upper_bound} and a parameter $\delta_1$ to be discussed shortly, we set
\begin{equation}
\label{eqn:upper_bound}
    \sigma^{u}(S^o) = \Bigg ( \sqrt{ \frac{\Lambda_1(S^*)}{1 - 1/e} + \ln(1 / \delta_1)} + \sqrt{\ln(1 / \delta_1)} \Bigg )^2 \frac{n}{\theta_1 (1-\epsilon')}
\end{equation}

\stitle{Putting It Together.} The reason \emph{NAMM} ensures a $(1 - 1/e - \epsilon)$ approximation with at least $(1 - \delta)$ probability is as follows. First, the algorithm has at most $i_{max}$ iterations. In each of the first $i_{max} - 1$ iterations, Algorithm \ref{alg:namm} generates a size-$k$ seed set $S^*$ and derives $\sigma^l(S^*)$ and $\sigma^u(S^o)$ from $\mathcal{R}_2$ and $\mathcal{R}_1$, respectively, setting $\delta_1 = \delta_2 = \Delta/(3 i_{max})$. Then, it computes $\alpha \gets \sigma^{l}(S^{*}) / \sigma^{u}(S^{o})$ as the approximation guarantee of $S^*$. By Lemmas \ref{lem:lower_bound} and \ref{lem:upper_bound}, and conditioning on the event $(1-\epsilon') \leq \Gamma \leq (1+\epsilon')$, $\alpha$ is correct with at least $1 - 2 \Delta / (3 i_{max})$ probability. By the union bound, it has at most $\frac{2 \Delta}{3}$ probability to output an incorrect solution in the first $i_{max}-1$ iterations. Meanwhile, in the last iteration, it returns a seed set $S^*$ generated by applying the greedy algorithm on $\mathcal{R}_1$, with $|\mathcal{R}_1| \geq N_{max}$. By Equation \ref{eqn:n_max} and conditioning on the event $(1-\epsilon') \leq \Gamma \leq (1+\epsilon')$, this ensures that $S^*$ is an $(1-1/e-\epsilon)$-approximation with at least $1 - \Delta/3$ probability. Therefore, the probability that \emph{NAMM} returns an incorrect solution is at most $\frac{2 \Delta}{3} + \Delta/3 + \delta' = \delta$ leading to the following result regarding Algorithm \ref{alg:namm}.
\begin{thm}
Algorithm \ref{alg:namm} returns a $(1-1/e-\epsilon)$-approximate solution for $\underline\mu$ as well as $\overline\mu$ with at least $1 - \delta$ probability.
\end{thm}

\stitle{Runtime.} \edit{The computation overhead of \emph{NAMM} consists of (i) the generation of RDR sets, (ii) the execution of greedy, and (iii) the computation of $\sigma^{l}(S^{*})$ and $\sigma^{u}(S^{o})$. Both (ii) and (iii) are linear in the number of RDR sets, therefore the time complexity of \emph{NAMM} is $O(\sum_{R \in \mathcal{R}_1 \cup \mathcal{R}_2} EPT)$, where $EPT$ is the expected runtime required to generate an RDR set. Next, we compare $EPT$ to the expected runtime required to generate a traditional RR set used for the IM problem allowing us to gauge the time complexity of the MM problem relative to the IM problem. Let $INF_1$ be the largest expected influence of any size-1 node set in $G$ under the IC model. We have the following lemma.}
\begin{lem}
\label{lem:ept}
The expected runtime to generate a RDR set is given by $EPT = O \Big ( INF_F + \frac{m}{n} \cdot \Big ( INF_1 \Big )^2 \Big )$.
\end{lem}
\begin{proof}
Define $n^{*}$ and $m^{*}$ as the number of nodes and edges inspected by Phase II of the RDR construction algorithm. Phase I is an invocation of BFS, and has expected runtime $INF_F$. The Phase II traversal has a runtime of $n^{*} \log n^{*} + m^{*}$. Finally, Phase III applies a dynamic programming routine on $\mathcal{A}_2$.
Further, Phase III applies the dynamic programming routine on at most every node in $\mathcal{S}$, and so the runtime is bounded by $n^{*} m^{*}$. In total, the runtime is bounded by
\begin{equation}
\label{eqn:ept}
    EPT = EPT_I + EPT_{II} + EPT_{III}
\end{equation}
Next, we make the observation that the size of $\mathbb{E}[|\mathcal{S}|] = \mathbb{E}[n^{*}] \leq OPT_{inf}^{1}$. This follows from the definition of $OPT_{inf}^{1}$. Further, $\mathbb{E}[m^{*}] \leq \frac{m}{n} \cdot OPT_{inf}^{1}$. This follows from known results for the IM problem under the IC model \cite{tang2015influence}. Combining \ref{eqn:ept} and the above observations, we get
\begin{align*}
    EPT &= INF_F + \mathbb{E}[m^* + n^* \log n^* + n^* m^*] \\
        &\leq INF_F + \mathbb{E}[2 n^* m^*] \\
        &= O \Big ( INF_F + \frac{m}{n} \Big ( OPT_{inf}^{*} \Big )^2 \Big )
\end{align*}
\end{proof}
In contrast, the expected runtime to generate a traditional RR set is $EPT_{RR} = O(\frac{m}{n} \cdot INF_1)$. Thus, in terms of $EPT_{RR}$, the expected runtime to generate an RDR set is $O(INF_F + EPT_{RR} \cdot INF_1)$. The additional factor of $INF_1$ in the runtime expression is explained by observing that RDR generation has to potentially perform tie-breaking on every node reverse-reachable from the root of the RDR set. Meanwhile, the additional $INF_F$ term follows from the necessary computation of the set of nodes influenced by $F$.

\stitle{Sandwich Algorithm.} Based on the above results, the \emph{SA} algorithm is given in Algorithm \ref{alg:sandwich} and returns the seed set $S^{*} \in \{ S^L, S^U \}$ that leads to the largest objective value w.r.t.\ the original objective. The solution produced by \emph{SA} has the following performance guarantee due to \cite{lu2015competition}.
\begin{lem}[\cite{lu2015competition}]
\label{lem:sandwich}
$\mu(S^{*}) \ge \beta \cdot (1 - 1/e) \cdot \mu(S^o)$, where \sloppy $\beta = \max \big \{ \frac{\mu(S^U)}{\overline\mu(S^U)}, \frac{\underline\mu(S^o)}{\mu(S^o)} \big \} $ and $S^o$ is the optimal solution to the MM problem.
\end{lem}
\vspace{-10pt}
\begin{algorithm}
\caption{Sandwich Approximation($\underline\mu$, $\overline\mu$, $G$, $k$)}\label{alg:sandwich}
\begin{algorithmic}[1]
\State $S^L \gets NAMM(\underline\mu, G, k)$
\State $S^U \gets NAMM(\overline\mu, G, k)$
\State \textbf{return} $S^{*} = \argmax_{S \in \{S^L, S^U\}} \mu(S)$
\end{algorithmic}
\end{algorithm}

\subsection{Anytime Algorithm}
Often, in MM applications, it is advantageous to have the ability to know how close to the optimum the current solution at any time is, so we can stop if desired. The OPIM algorithm developed for IM is such an anytime algorithm. Unlike IM, the non-submodular objective of MM complicates the development of an anytime algorithm. We overcome this challenge as follows. 

Lemmas \ref{lem:lower_bound} and \ref{lem:upper_bound} are the key ingredients required for an anytime algorithm for the MM problem. We can modify Algorithm \ref{alg:namm} to stop at a user-specified timestamp and produce a seed set along with its current approximation assurance. However, due to the non-submodular behaviour of the mitigation objective, the anytime algorithm has to generate RDR sets for both the upper and lower bounding functions \textit{simultaneously} in order to apply Algorithm \ref{alg:sandwich} on the resulting anytime solutions. Formally, for each of the submodular bounding functions, given the collections $\mathcal{R}_1$ and $\mathcal{R}_2$ of RDR sets that have been generated so far, our anytime algorithm derives a size-$k$ seed set $S^*$ and derives $\sigma^l(S^*)$ and $\sigma^u(S^o)$ from $\mathcal{R}_2$ and $\mathcal{R}_1$, respectively, and returns $\alpha \gets \sigma^{l}(S^{*}) / \sigma^{u}(S^{o})$ as the approximation guarantee of $S^*$. By Lemmas \ref{lem:lower_bound} and \ref{lem:upper_bound}, $\alpha$ is correct with at least $1 - \delta_1 - \delta_2$ probability. Therefore, setting $\delta_1 = \delta_2 = \delta / 2$ suffices to ensure the failure probability does not exceed $\delta$.
\section{RDR Set Generation}
\label{sec:rdr_gen}

In this section we describe our sampling process for constructing RDR sets for both $\underline\mu(\cdot)$ and $\overline\mu(\cdot)$. We will begin by describing how to compute an RDR set in a possible world $X$ corresponding to the lower bounding objective. We will follow up with a modification for the upper bounding objective  which handles the extended possible worlds $X'$ required by the latter.

Constructing RDR sets is significantly more difficult than traditional RR sets due to the fact that the membership status and associated weight of a node in the RDR set depends on the complex propagation interactions between campaigns $F$ and $M$. This interaction requires that we generate RDR sets with a particular order of edge propagation in mind and requires that we maintain a more nuanced notion of membership. The nature of our problem requires that we identify nodes with the potential to be reached by the mitigation campaign within the node's activation window.

Furthermore, to the best of our knowledge, we are the first to consider how to incorporate a weighted random tie-breaking policy among reverse sampling approaches. Importantly, we must overcome the inherent difficulty that to be able to resolve a tie-break we need to have complete information about \emph{all} shortest paths from $S_F$ to the random RDR root node $v$ \emph{and} from an RDR set candidate $u$ to the random RDR root node $v$ in a possible world. We develop a dynamic programming approach to overcome this challenge.

\stitle{Resolving Tie-breaks.} Let $X$ be a random possible world and $v$ a randomly chosen root node of an RDR set. Denote the set of nodes that can reach $v$ in $X$ as $\mathcal{S}$ and consider a candidate RDR set node $u \in \mathcal{S}$. In the possible world $X$, $v$ becomes active in campaign $M$ due to $u$ if there exists a path $P$ from $u$ to $v$ for which all tie-breaks along $P$ are won by $M$. Next, we describe how to recursively compute the event that all the tie-breaks required for $v$ to adopt campaign $M$ when using $u$ as a seed node are successful via a dynamic programming procedure that captures the dynamics of the \model model unfolding in $X$.


We denote the event that a node $w \in \mathcal{S}$ adopts either campaign $C \in \{M,F\}$ in $X$ when using $u$ as a seed by the indicator $a(w,C) : C \rightarrow \{0,1\}$. Also, we denote the timestep in which $w$'s AW window closes in $X$ when using $u$ as a seed as $close(w)$. Our tie-breaking algorithm is presented in Algorithm \ref{alg:tb}.
\begin{algorithm}
\caption{Tie Breaking}\label{alg:tb}
\algorithmicrequire $G$, $S_F$, $u$, $\mathcal{S}$
\begin{algorithmic}[1]
\State Initialize priority queue $Q \gets \emptyset$
\ForAll{$w \in N^{+}(S_F)$}
    \State If $w \not\in Q$, insert $<w,1+\tau_w>$ into $Q$
\EndFor
\ForAll{$w \in N^{+}(u) \setminus N^{+}(S_F)$}
    \State If $w \not\in Q$, insert $<w,h_{(u,w)}+\tau_w>$ into $Q$
\EndFor
\State $a(w,M) = 1$ if $w = u$ and $0$ otherwise
\State $a(w,F) = 1$ if $w \in S_F$ and $0$ otherwise
\While{$Q$ is not empty}
    \State $<w,t> \gets Q.pop()$
    \State $close(w) = t$
    \ForAll{$w' \in \pi(w)$}
        \If{$a(w',M)$ and $close(w') + h_{(w',w)} \leq t$}
            \State $a(w,M) = 1$, \textbf{break}
        \ElsIf{$a(w',F)$ and $close(w') + 1 \leq t$}
            \State $a(w,F) = 1$, \textbf{break}
        \EndIf
    \EndFor
    \ForAll{$w' \in N^{+}(w)$}
        \State $t' \gets t + \tau_{w'} + a(w,M) \cdot h_{(w,w')} + a(w,F)$
        \If{$w' \not\in Q$}
            \State Insert $<w', t'>$ into $Q$
        \Else
            \State If $t' < priority(w')$ then $priority(w') \gets t'$
        \EndIf
    \EndFor
\EndWhile
\end{algorithmic}
\end{algorithm}
Intuitively, for each node $w$ processed in Algorithm \ref{alg:tb}, an active neighbour $w'$ of $w \in \mathcal{S}$ is chosen uniformly at random and $w$ activates in the same campaign as $w'$. The status of each node in $\mathcal{S}$ is computed recursively by traversing from nodes in $\{u\} \cup S_F$ to the root $v$. Define $\mathcal{P}_F$ ($\mathcal{P}_u$) as the collection of paths from $S_F$ ($u$) to $v$ in $X$. The time complexity and correctness of Algorithm \ref{alg:tb} is established in the following theorem.
\begin{thm}
Given a candidate seed node $u$, misinformation seeds $S_F$, and a random RDR set root $v$, Algorithm \ref{alg:tb} correctly computes $a(v,F)$ and $a(v,M)$ and runs in time $O(|\mathcal{P}_{F}| + |\mathcal{P}_{u}|)$.
\end{thm}

\subsection{Construction Algorithm}
We are now ready to describe our RDR set generation algorithm which works in three phases: (1) a forward traversal from $S_F$ to determine $R_F^X$, for choosing the RDR set $R$'s root node $v$, (2) a backwards traversal from $v$ to identify candidate RDR set nodes, and (3) A forward dynamic programming routine to resolve tie-breaks. Note, we can employ \emph{lazy sampling} by only revealing edge ``liveness'', activation window lengths, and meeting events on demand, based on the principle of deferred decisions. Thus, we avoid the costly computation of sampling the entire possible world. We include a running example on the graph in Figure \ref{fig:rdr_gen_1}.
\begin{figure}
\centering
\begin{tikzpicture}

\tikzset{node/.style={circle,draw,minimum size=0.5cm,inner sep=0pt},}
\tikzset{edge/.style={->,> = latex'},}

\node[node] (1) {$v_F$};
\node[node] (2) [right = 1cm of 1] {$v_1$};
\node[node] (3) [right = 1cm of 2] {$v_2$};
\node[node] (4) [below = 0.5cm of 1] {$v_3$};
\node[node] (5) [right = 1cm of 4] {$v_4$};
\node[node] (6) [below = 0.5cm of 4] {$v_6$};
\node[node] (7) [below = 0.5cm of 5] {$v_7$};
\node[node] (8) [right = 1cm of 7] {$v_8$};
\node[node] (9) [right = 1cm of 5] {$v_5$};

\path[draw,thick,->,> = latex']
(1) edge node {} (2)
(1) edge node {} (4)
(2) edge node {} (5)
(2) edge node {} (3)
(4) edge node {} (5)
(3) edge node {} (5)
(6) edge node {} (4)
(6) edge node {} (7)
(7) edge node {} (5)
(9) edge node {} (5)
(8) edge node {} (7);
\end{tikzpicture}
\caption{Example graph.}
\label{fig:rdr_gen_1}
\end{figure}
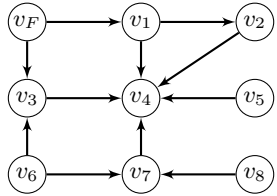

\stitle{Phase I: Forward BFS.} First, we run a breadth-first search (BFS) from $S_F$ in $G$. Since we are only interested in determining $R_F^X$ in the Phase I traversal, we do not incorporate AW's or meeting events during the BFS. The stopping condition for the traversal is when the queue of ``live'' edges to traverse are exhausted.

Next, as per our IS framework, we select a node uniformly at random from $R_F^X$ as the root of $R$. The collection of paths $\mathcal{P}_F$ from $S_F$ to $v$, to be leveraged in subsequent phases, can be constructed by following parent relationships along live edges from $v$ to nodes in $S_F$. Finally, we can determine the distance $d_X(S_F,v)$ during the path reconstruction by a simple dynamic programming procedure. Specifically, we sample AW lengths along the paths in $\mathcal{P}_F$ and retain the minimum value seen at each node in the reconstruction. We maintain the state, denoted $\mathcal{A}_1$, discovered in Phase I as defined by $\mathcal{P}_F$. Specifically, we maintain the traversal DAG from $S_F$ to $v$ defined on $X$ along with the sampled AW lengths for paths in $\mathcal{P}_F$.

In Figure \ref{fig:rdr_gen_2} we show an example traversal from $v_F$ where solid lines represent live edges and dotted lines represent dead edges. The graph represents the collection of reachable nodes $R_F^X$. Suppose we randomly select $v$ as the RDR set root from the collection $R_F^X$. Finally, the DAG represented by the subgraph connected by red edges, corresponds to the collection of paths $\mathcal{P}_F$ and is maintained, along with the sampled AW lengths shown in red, as the state $\mathcal{A}_1$. The distance from $v_F$ to $v$ is determined to be $d(v_F,v) = 3$.
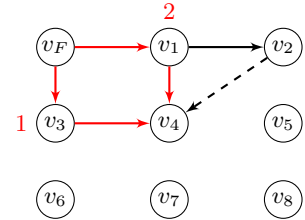
\begin{figure}
\centering
\begin{tikzpicture}

\tikzset{node/.style={circle,draw,minimum size=0.5cm,inner sep=0pt},}
\tikzset{edge/.style={->,> = latex'},}

\node[node] (1) {$v_F$};
\node[node, label={[red]above:$2$}] (2) [right = 1cm of 1] {$v_1$};
\node[node] (3) [right = 1cm of 2] {$v_2$};
\node[node, label={[red]left:$1$}] (4) [below = 0.5cm of 1] {$v_3$};
\node[node] (5) [right = 1cm of 4] {$v_4$};
\node[node] (6) [below = 0.5cm of 4] {$v_6$};
\node[node] (7) [below = 0.5cm of 5] {$v_7$};
\node[node] (8) [right = 1cm of 7] {$v_8$};
\node[node] (9) [right = 1cm of 5] {$v_5$};

\path[draw,thick,->,> = latex']
(1) edge [red] node {} (2)
(1) edge [red] node {} (4)
(2) edge [red] node {} (5)
(2) edge node {} (3)
(4) edge [red] node {} (5)
(3) edge [dashed] node {} (5);
\end{tikzpicture}
\caption{Phase I execution. Dashed lines represent edges that are not live. The subgraph spanned by the red edges makes up the Phase I traversal DAG from $v_F$ to $v_4$. Red numbers next to nodes represent sampled AW lengths. The distance from $v_F$ to $v_4$ is determined to be $d(v_F,v_4) = 3$.}
\label{fig:rdr_gen_2}
\end{figure}

\stitle{Phase II: Backwards Delayed-distance Djikstra.} Phase II generates two disjoint components of $R$ corresponding to the set of candidate RDR nodes that require tie-breaking and those that do not. To begin, we run an invocation of Djikstra backwards from $v$. The nodes encountered during this traversal are candidates to be included in $R$. Notice that some edges visited will already have their liveness determined from Phase I while the rest have their liveness sampled as necessary. Further, lazy sampling reveals results of meeting events on demand. In particular, when traversing a live edge $e \in X$, the lazy generation of meeting events for $e$ using $m(e)$ can be reduced to sampling a value for $h_e$ from a geometric distribution. This approach allows us to bypass flipping a coin with bias $m(e)$ for each live edge $e$ in every iteration until a meeting event occurs. Unlike Phase I, we maintain delayed-distances for every node reached in Phase II. Delayed-distances for neighbours of $v$ are computed as $dd_X(v,w) = h_{(v,w)}$ for $w \in N^{-}(v)$ and otherwise, when processing an edge $e = (w,w')$, we can compute the delayed-distance of $w'$ as $dd_X(v, w') = dd_X(v,w) + \tau_w + h_e$. Similar to Phase I, we also maintain the state $\mathcal{A}_2$ discovered in Phase II. In particular, $\mathcal{A}_2$ corresponds to the traversal DAG generated by the Phase II traversal in $X^T$, i.e. the set of live in-edges for each node reached in addition to the sampled meeting events and AW length parameters.

Next we describe some additional bookkeeping performed in Phase II. We maintain an \emph{overlap} indicator variable $\mathbb{I}_{w}^{OL}$ at each node $w$ reached in Phase II that determines if any path from $w$ to $v$ contains a node $w'$ reached by $F$. $\mathbb{I}_{w}^{OL}$ is calculated during the traversal by setting the value to $1$ if (i) $w' \in \mathcal{A}_1$ for some $w' \neq v \in N^{-}(w)$ or (ii) $\mathbb{I}_{w'}^{OL} = 1$ for some $w' \neq v \in N^{-}(w)$, and $0$ otherwise. If $\mathbb{I}_{w}^{OL} = 0$ then the collection of paths from $S_F$ to $v$ excluding $w$ are disjoint from the collection of paths from $w$ to $v$. As a result, calculating the reward associated with $w$ in the RDR set is completely determined by delayed-distance values as there are no potential tie-breaks to be resolved.

The following type of nodes do not require tie-breaking. First, nodes $u$ at a delayed-distance less than $d_X(S_F,v) - \tau_v$ that satisfy $\mathcal{I}_{u}^{OL} = 0$. Intuitively, the campaign $M$ from these nodes reaches $v$ early enough that the AW closes before $F$ arrives. All such nodes are used to construct a pair $(u, 2)$ and inserted into $R$, i.e., $M$ arrives at these nodes well before $F$ so the $M$ campaign gets a maximum credit of $2$ on these nodes.  Second, nodes $u$ at a delayed-distance less than or equal to $d_X(S_F,v) + \tau_v$ that satisfy $\mathcal{I}_{u}^{OL} = 0$ also achieve a reward that is completely determined by their delayed-distance to $v$, so we can construct a pair $(u, 1)$ and insert it into $R$. All nodes that are not of the above two types are inserted into a set $\mathcal{S}_{OL}$ for tie-breaking in Phase III.

In Figure \ref{fig:rdr_gen_3} we show an example Phase II traversal from $v$ where solid lines represent live edges and dotted lines represent dead edges. Blue numbers next the edges represent the sampled meeting lengths. Table \ref{tbl:rdr_gen_3} shows the delayed distances and overlap indicator values resulting from Phase II.
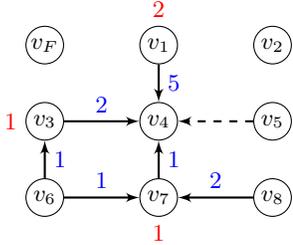
\begin{figure}
\centering
\begin{tikzpicture}

\tikzset{node/.style={circle,draw,minimum size=0.5cm,inner sep=0pt},}
\tikzset{edge/.style={->,> = latex'},}

\node[node] (1) {$v_F$};
\node[node, label={[red]above:$2$}] (2) [right = 1cm of 1] {$v_1$};
\node[node] (3) [right = 1cm of 2] {$v_2$};
\node[node, label={[red]left:$1$}] (4) [below = 0.5cm of 1] {$v_3$};
\node[node] (5) [right = 1cm of 4] {$v_4$};
\node[node] (6) [below = 0.5cm of 4] {$v_6$};
\node[node, label={[red]below:$1$}] (7) [below = 0.5cm of 5] {$v_7$};
\node[node] (8) [right = 1cm of 7] {$v_8$};
\node[node] (9) [right = 1cm of 5] {$v_5$};

\path[draw,thick,->,> = latex']
(2) edge node[right, blue] {5} (5)
(4) edge node[above, blue] {2} (5)
(6) edge node[right, blue] {1} (4)
(6) edge node[above, blue] {1} (7)
(7) edge node[right, blue] {1} (5)
(8) edge node[above, blue] {2} (7)
(9) edge [dashed] node {} (5);
\end{tikzpicture}
\caption{Phase II execution. Blue numbers next the edges represent the sampled meeting lengths and an additional AW length was sampled.}
\label{fig:rdr_gen_3}
\end{figure}
\begin{table}
\centering
\begin{tabular}{ | c | c | c |}
    \hline
    Node & delayed dist & overlap \\ \hline \hline
    $v_1$ & 5 & 0 \\ \hline
    $v_3$ & 2 & 0 \\ \hline
    $v_6$ & 3 & 1 \\ \hline
    $v_7$ & 1 & 0 \\ \hline
    $v_8$ & 4 & 0 \\ \hline
   \end{tabular}
    \caption{Phase II details.}
    \label{tbl:rdr_gen_3}
\end{table}


\stitle{Phase III: Forward Tie-breaking.} Phase III proceeds by resolving tie-breaks for all nodes in $\mathcal{S}_{OL}$ via the dynamic programming approach described by Algorithm \ref{alg:tb}. Recall, the input to Algorithm \ref{alg:tb} requires the collection of paths defined by $\mathcal{P}_F$ and $\mathcal{P}_u$ in order to determine the reward achieved when using $u$ as a seed node for campaign $M$. Observe, the state collected from Phases I and II are sufficient to construct $\mathcal{P}_F \cup \mathcal{P}_u$ for all $u \in \mathcal{S}_{OL}$ since each such collection of paths is a subgraph of the union of the DAGs constructed during Phase I and II, i.e.\ $\mathcal{A}_1 \cup \mathcal{A}_2$.


For a node $u \in \mathcal{S}_{OL}$, if $a(v,M) = 1$ and every neighbour of $v$, $v' \in N^{-}(v)$ satisfies $a(v',M) = 1$ then we construct the pair $(u, 2)$. If $a(v,M) = 1$ but some neighbour of $v$, $v' \in N^{-}(v)$ satisfies $a(v',F) = 1$ and $\mathbb{I}_{v',v}^{F} = 1$ then we construct the pair $(u, 1)$. Finally, if $a(v,M) = 0$ but some neighbour of $v$, $v' \in N^{-}(v)$ satisfies $a(v',M) = 1$ and $\mathbb{I}_{v',v}^{M} = 1$ then we construct the pair $(u, 1)$. Otherwise, $u$ cannot reach $v$ within its AW so we do not add a tuple to $R$.
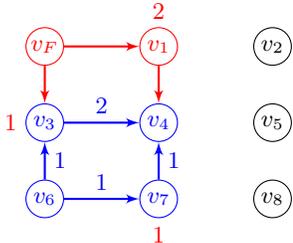
\begin{figure}
\centering
\begin{tikzpicture}

\tikzset{node/.style={circle,draw,minimum size=0.5cm,inner sep=0pt},}
\tikzset{edge/.style={->,> = latex'},}

\node[node, red] (1) {$v_F$};
\node[node, red, label={[red]above:$2$}] (2) [right = 1cm of 1] {$v_1$};
\node[node] (3) [right = 1cm of 2] {$v_2$};
\node[node, blue,  label={[red]left:$1$}] (4) [below = 0.5cm of 1] {$v_3$};
\node[node, blue] (5) [right = 1cm of 4] {$v_4$};
\node[node, blue] (6) [below = 0.5cm of 4] {$v_6$};
\node[node, blue, label={[red]below:$1$}] (7) [below = 0.5cm of 5] {$v_7$};
\node[node] (8) [right = 1cm of 7] {$v_8$};
\node[node] (9) [right = 1cm of 5] {$v_5$};

\path[draw,thick,->,> = latex']
(1) edge [red] node {} (2)
(1) edge [red] node {} (4)
(2) edge [red] node {} (5)
(4) edge [blue] node[above, blue] {2} (5)
(6) edge [blue] node[right, blue] {1} (4)
(6) edge [blue] node[above, blue] {1} (7)
(7) edge [blue] node[right, blue] {1} (5);
\end{tikzpicture}
\caption{Phase III execution for node $v_6$ with $\tau_{v_4} = 2$. Red nodes have adopted $F$ and blue nodes have adopted $M$. Campaign $F$ ($M$) propagates along red (blue) edges.}
\label{fig:rdr_gen_4}
\end{figure}

\section{Experiments}
\label{sec:experiments}
\begin{table}
\small
\caption{Network Statistics.}
\centering
\begin{tabular}{|l|l|l|l|l|l|}
\hline
& Gnutella & Flixster & D-B & D-M & DBLP \\ \hline
\# nodes & 6.3K & 7.6K & 23.3K & 34.9K & 317K \\ \hline
\# edges & 20.8K & 75.7K & 141K & 274K & 2.1M \\ \hline
avg. deg & 3.3 & 9.43 & 6.5 & 7.9 & 6.62 \\ \hline
type & dir & undir & dir & dir & undir \\ \hline
\end{tabular}
\label{tbl:networks}
\end{table}
\stitle{Setup.} We perform experiments on 5 real networks. 
The seeds for campaign $F$ are generated in two ways. First, we sample a small number of nodes from the top-$k$ most influential nodes in the network to simulate the spread of fake news by a popular user in the network. Second, we sample a larger number of users at random from the network to simulate the coordinated spread of misinformation by bots or newly created puppet accounts. \edit{Figures for each network list the results for influential fake seeds (left) and random fake seeds (right).} In our experiments, the mitigation achieved by the final solution of each algorithm is evaluated by 20K Monte Carlo simulations. All experiments are performed on a Linux machine with Intel Xeon 2.6 GHz CPU and 128 GB RAM.

\stitle{Networks.} Table \ref{tbl:networks} summarizes the networks and their characteristics. Flixster is mined from a social movie site and a strongly connected component is extracted. Gnutella is a peer-to-peer file sharing network. Douban is a Chinese social network, where users rate books, movies, music, etc. All movie and book ratings of the users in the graph are crawled separately to derive two datasets from book and movie ratings: Douban-Book and Douban-Movie. Finally DBLP, a peer collaboration network, is a large network that we use to test scalability.

\stitle{Baselines.} Since there is no algorithm explicitly addressing the model and reward function considered in this paper, we compare the mitigation achieved by \emph{NAMM} against the following baselines. \edit{\emph{IMM} \cite{tang2015influence} is a state-of-the-art solution for the IM problem.} The \textsc{Influential} baseline selects seeds in decreasing order of expected influence. The \textsc{Proximity} baseline selects seed nodes from the out-neighbors of the fake seeds, where a preference is given to those nodes connected by a high probability edge. \textsc{Random} is a baseline method which selects the seeds randomly.


\stitle{Default Parameters.} Following previous works \cite{tong2018misinformation, tong2019beyond, simpson2020reverse, tong2017efficient, tang2015influence, tang2018online, lu2015competition} we set edge probabilities for $e = (u,v)$ to $1 / indeg(v)$, where $indeg(v)$ is the in degree of node $v$. Unless otherwise specified, we use $\epsilon = 0.1$ and $\delta = 1 / n$ as our default for all methods.

\underline{\etitle{Meeting Probabiltiies.}} \edit{In the same way as \cite{chen2012time}, where meeting events were introduced, we sample meeting lengths from a geometric distribution. We calibrate the meeting probabilities for campaign $M$ based on observations made in \cite{vosoughi2018spread}. The authors investigate the temporal characteristics of the propagation of true and false news over all the fact-checked cascades that spread on Twitter from 2006 to 2017 totalling over 126,000 cascades. The classification into truth or falsehood was established by six independent fact-checking organizations. The resulting cascades were compared on depth (number of retweet hops from the source tweet), size (number of users involved in the cascade), maximum breadth, and structural virality. The authors observed that falsehood diffused significantly farther, faster, deeper, and more broadly than the truth. In particular, it was observed that truth rarely reached more than 1500 users and it took the truth about six times as long as falsehood to reach 1500 people. Based on these observations, the meeting length distribution is parameterized by success probability $m(e) = 1 / 6$ so that, on average, the misinformation propagates 6$\times$ faster than the truth.}


\underline{\etitle{Activation Windows.}} \edit{The activation window lengths are generated by a two-step procedure. First, we simulate if a user reads the linked content following observations in \cite{gabielkov2016social} on real-world click-through behaviours. The authors present a large scale, unbiased study of social clicks by gathering a month of web visits to online resources mentioned in Twitter. Their dataset covers 2.8 million shares and 9.6 million actual clicks. The authors estimate that a majority ($59\%$) of URLs mentioned on Twitter are not clicked at all. Informed by these observations, we first flip a biased coin with probability $0.6$ and set the AW length to $0$ if the flip succeeds.} 

\edit{Meanwhile, non-zero AW lengths are generated based on observations made in \cite{mitchell2016long} on real-world reading times of social media users. The authors utilize audience behavior metrics provided by the web analytics firm Parse.ly covering 117 million anonymized, complete cellphone interactions with 74,840 articles from 30 news websites in the month of September 2015. Overall, short-form news stories ($<1000$ words) represent 76\% of the articles while long-form articles ($\geq 1000$ words) account for the remaining 24\% of articles. The authors observe that on short-form and long-form articles users spend on average 57 seconds and 123 seconds reading, respectively. Further, a more fine-grained breakdown of average reading times reveals a distribution that closely resembles a geometric distribution. Based on these observations, we generate a sample distribution of reading times from a geometric distribution parameterized by $p=1/57$ with probability $0.76$ and $p=1/123$ with probability $0.24$ reflecting the distribution of short-form and long-form articles. Finally, we apply the bias-corrected maximum likelihood estimator for geometric distributions \cite{cox1968general} to learn a AW length parameterization of $74$ seconds.}


\edit{Finally, AW lengths must be normalized to the length of a single hop in our model, which corresponds to the base propagation rate of misinformation measured in seconds. We learned a normalization factor from a collection of retweet cascades of misinformation crawled from Twitter during Oct. 10--Nov. 10, 2020. We extract the distribution of retweet intervals from the cascades and clean the distribution by removing the effect of bots that are programmed to instantly retweet particular accounts, by removing intervals of length less than $3$ seconds. We observed that the resulting distribution closely resembles a geometric distribution. As a result, we apply the bias-corrected maximum likelihood estimator for geometric distributions \cite{cox1968general} to learn a base propagation rate of $200$ seconds for campaign $F$.}




\begin{figure*}
\centering
\begin{tikzpicture}[scale=0.3]
\pgfplotsset{every axis legend/.append style={
        at={(0.7,1.07)},
        anchor=south}}
\begin{axis}[
    font=\LARGE,
    ylabel={Mitigation},
    xlabel={$k$},
    xmin=1, xmax=20,
    ymin=0, ymax=2200,
    xtick={1,3,5,10,15,20},
    ymajorgrids=true,
    grid style=dashed,
    legend columns=5,
]

\addplot[
    color=blue,
    mark=square,
]
coordinates { (1,233.25) (3,623.72) (5,954.34) (10,1645.79) (15,1994.59) (20,2179.26) };
\addlegendentry{\emph{NAMM}}

\addplot[
    color=cyan,
    mark=triangle,
]
coordinates { (1,60.82) (3,291.27) (5,421.49) (10,709.03) (15,880.54) (20,984.74) };
\addlegendentry{\emph{IMM}}
    
\addplot[
    color=red,
    mark=triangle,
]
coordinates { (1,50.60) (3,221.44) (5,270.10) (10,432.10) (15,818.17) (20,1086.59) };
\addlegendentry{\textsc{Inf}}

\addplot[
    color=purple,
    mark=diamond,
]
coordinates { (1,2.00) (3,110.19) (5,230.54) (10,499.82) (15,698.08) (20,907.95) };
\addlegendentry{\textsc{Prox}}

\addplot[
    color=orange,
    mark=x,
]
coordinates { (1,3.67) (3,21.21) (5,10.85) (10,30.36) (15,72.27) (20,54.70) };
\addlegendentry{\textsc{Ran}}
    
\end{axis}
\end{tikzpicture}
\hspace{-7.5ex}
\begin{tikzpicture}[scale=0.3]
\begin{axis}[
    font=\LARGE,
    ylabel={Mitigation},
    xlabel={$k$},
    xmin=1, xmax=20,
    ymin=0, ymax=440,
    xtick={1,3,5,10,15,20},
    ymajorgrids=true,
    grid style=dashed,
]

\addplot[
    color=blue,
    mark=square,
]
coordinates { (1,75.67) (3,175.40) (5,219.70) (10,323.56) (15,395.28) (20,439.82) };

\addplot[
    color=cyan,
    mark=triangle,
]
coordinates { (1,25.21) (3,64.25) (5,83.09) (10,127.99) (15,172.94) (20,198.14) };
    
\addplot[
    color=red,
    mark=triangle,
]
coordinates { (1,17.21) (3,33.00) (5,47.15) (10,68.73) (15,101.24) (20,129.66) };

\addplot[
    color=purple,
    mark=diamond,
]
coordinates { (1,16.98) (3,21.15) (5,25.39) (10,110.35) (15,122.13) (20,132.00) };

\addplot[
    color=orange,
    mark=x,
]
coordinates { (1,1.02) (3,6.88) (5,10.07) (10,15.48) (15,18.43) (20,19.92) };
    
\end{axis}
\end{tikzpicture}
\begin{tikzpicture}[scale=0.3]
\begin{axis}[
    font=\LARGE,
    ylabel={Mitigation},
    xlabel={$k$},
    xmin=1, xmax=20,
    ymin=0, ymax=620,
    xtick={1,3,5,10,15,20},
    ymajorgrids=true,
    grid style=dashed,
]

\addplot[
    color=blue,
    mark=square,
]
coordinates { (1,72.89) (3,188.87) (5,266.18) (10,414.24) (15,522.55) (20,616.45) };

\addplot[
    color=cyan,
    mark=triangle,
]
coordinates { (1,43.27) (3,117.42) (5,202.48) (10,309.83) (15,388.98) (20,430.50) };
    
\addplot[
    color=red,
    mark=triangle,
]
coordinates { (1,33.41) (3,94.78) (5,154.12) (10,247.86) (15,334.67) (20,382.71) };

\addplot[
    color=purple,
    mark=diamond,
]
coordinates { (1,10.70) (3,17.29) (5,21.30) (10,34.48) (15,55.78) (20,151.61) };

\addplot[
    color=orange,
    mark=x,
]
coordinates { (1,0.43) (3,8.59) (5,5.46) (10,6.69) (15,25.05) (20,29.60) };
    
\end{axis}
\end{tikzpicture}
\begin{tikzpicture}[scale=0.3]
\begin{axis}[
    font=\LARGE,
    ylabel={Mitigation},
    xlabel={$k$},
    xmin=1, xmax=20,
    ymin=0, ymax=165,
    xtick={1,3,5,10,15,20},
    ymajorgrids=true,
    grid style=dashed,
]

\addplot[
    color=blue,
    mark=square,
]
coordinates { (1,24.85) (3,52.57) (5,74.83) (10,115.95) (15,139.54) (20,159.03) };

\addplot[
    color=cyan,
    mark=triangle,
]
coordinates { (1,11.70) (3,28.59) (5,50.47) (10,78.66) (15,103.47) (20,110.98) };
    
\addplot[
    color=red,
    mark=triangle,
]
coordinates { (1,7.40) (3,20.42) (5,33.83) (10,61.30) (15,69.93) (20,76.49) };

\addplot[
    color=purple,
    mark=diamond,
]
coordinates { (1,8.56) (3,12.55) (5,16.58) (10,26.57) (15,36.57) (20,70.57) };

\addplot[
    color=orange,
    mark=x,
]
coordinates { (1,0.43) (3,0.20) (5,2.25) (10,1.56) (15,5.32) (20,3.34) };
    
\end{axis}
\end{tikzpicture}
\begin{tikzpicture}[scale=0.3]
\begin{axis}[
    font=\LARGE,
    ylabel={$\beta$ lower bound},
    xlabel={$k$},
    xmin=1, xmax=20,
    ymin=0.6, ymax=0.9,
    xtick={1,3,5,10,15,20},
    legend pos=south east,
    ymajorgrids=true,
    grid style=dashed,
]

\addplot[
    color=teal,
    mark=square,
]
coordinates { (1,0.62) (3,0.68) (5,0.72) (10,0.78) (15,0.80) (20,0.81) };
\addlegendentry{FSTop}
    
\addplot[
    color=violet,
    mark=triangle,
]
coordinates { (1,0.71) (3,0.83) (5,0.82) (10,0.84) (15,0.85) (20,0.85) };
\addlegendentry{FSRan}
    
\end{axis}
\end{tikzpicture}
\begin{tikzpicture}[scale=0.3]
\begin{axis}[
    font=\LARGE,
    ylabel={$\beta$ lower bound},
    xlabel={$k$},
    xmin=1, xmax=20,
    ymin=0.66, ymax=0.9,
    xtick={1,3,5,10,15,20},
    ymajorgrids=true,
    grid style=dashed,
]

\addplot[
    color=teal,
    mark=square,
]
coordinates { (1,0.69) (3,0.70) (5,0.71) (10,0.72) (15,0.73) (20,0.75) };
    
\addplot[
    color=violet,
    mark=triangle,
]
coordinates { (1,0.73) (3,0.80) (5,0.83) (10,0.88) (15,0.86) (20,0.84) };
    
\end{axis}
\end{tikzpicture}
\caption{Mitigation on the Gnutella (left) \& Flixster (middle) datasets with $10$ ($50$) influential (random) fake seeds. Lower bound of $\beta$ for \emph{SA} (right).}
\label{fig:small_mit}
\end{figure*}
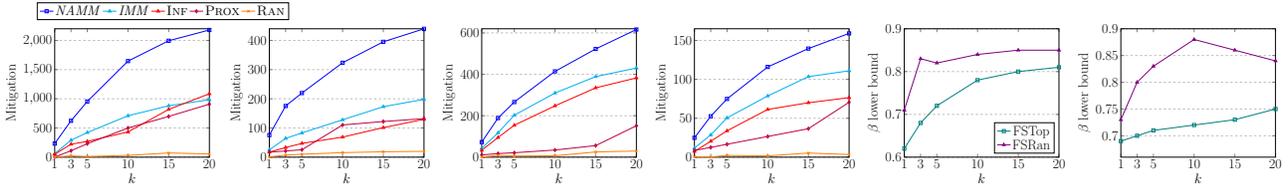



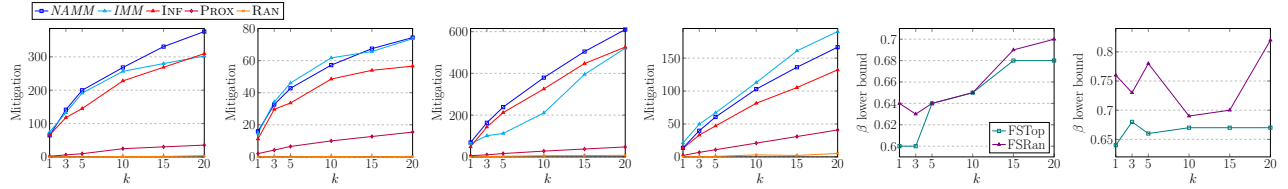
\begin{figure*}
\centering
    
\begin{tikzpicture}[scale=0.3]
\pgfplotsset{every axis legend/.append style={
        at={(0.7,1.07)},
        anchor=south}}
\begin{axis}[
    font = \LARGE,
    ylabel={Mitigation},
    xlabel={$k$},
    xmin=1, xmax=20,
    ymin=0, ymax=385,
    xtick={1,3,5,10,15,20},
    ymajorgrids=true,
    grid style=dashed,
    legend columns=5,
]

\addplot[
    color=blue,
    mark=square,
]
coordinates { (1,64.63) (3,141.45) (5,199.83) (10,268.15) (15,330.66) (20,375.80) };
\addlegendentry{\emph{NAMM}}

\addplot[
    color=cyan,
    mark=triangle,
]
coordinates { (1,74.56) (3,132.67) (5,191.36) (10,256.49) (15,279.91) (20,302.45) };
\addlegendentry{\emph{IMM}}
    
\addplot[
    color=red,
    mark=triangle,
]
coordinates { (1,65.57) (3,117.92) (5,145.00) (10,227.91) (15,268.73) (20,309.60) };
\addlegendentry{\textsc{Inf}}

\addplot[
    color=purple,
    mark=diamond,
]
coordinates { (1,2.00) (3,6.00) (5,9.51) (10,24.57) (15,30.14) (20,35.34) };
\addlegendentry{\textsc{Prox}}

\addplot[
    color=orange,
    mark=x,
]
coordinates { (1,0.01) (3,0.13) (5,0.25) (10,1.27) (15,1.33) (20,3.64) };
\addlegendentry{\textsc{Ran}}
    
\end{axis}
\end{tikzpicture}
\hspace{-7.5ex}
\begin{tikzpicture}[scale=0.3]
\begin{axis}[
    font= \LARGE,
    ylabel={Mitigation},
    xlabel={$k$},
    xmin=1, xmax=20,
    ymin=0, ymax=80,
    xtick={1,3,5,10,15,20},
    legend pos=north west,
    ymajorgrids=true,
    grid style=dashed,
]

\addplot[
    color=blue,
    mark=square,
]
coordinates { (1,15.94) (3,32.61) (5,42.77) (10,57.18) (15,67.45) (20,74.36) };

\addplot[
    color=cyan,
    mark=triangle,
]
coordinates { (1,14.72) (3,34.27) (5,46.04) (10,61.72) (15,65.68) (20,73.66) };
    
\addplot[
    color=red,
    mark=triangle,
]
coordinates { (1,11.17) (3,29.64) (5,33.59) (10,48.50) (15,53.94) (20,56.55) };

\addplot[
    color=purple,
    mark=diamond,
]
coordinates { (1,2.03) (3,4.25) (5,6.50) (10,9.97) (15,12.71) (20,15.44) };

\addplot[
    color=orange,
    mark=x,
]
coordinates { (1,0.13) (3,0.13) (5,0.13) (10,0.27) (15,0.10) (20,0.32) };
    
\end{axis}
\end{tikzpicture}
\begin{tikzpicture}[scale=0.3]
\begin{axis}[
    font= \LARGE,
    ylabel={Mitigation},
    xlabel={$k$},
    xmin=1, xmax=20,
    ymin=0, ymax=615,
    xtick={1,3,5,10,15,20},
    ymajorgrids=true,
    grid style=dashed,
]

\addplot[
    color=blue,
    mark=square,
]
coordinates { (1,69.20) (3,163.50) (5,238.50) (10,379.60) (15,504.66) (20,609.97) };

\addplot[
    color=cyan,
    mark=triangle,
]
coordinates { (1,62.96) (3,101.47) (5,112.41) (10,210.53) (15,395.12) (20,520.79) };
    
\addplot[
    color=red,
    mark=triangle,
]
coordinates { (1,48.39) (3,143.42) (5,212.28) (10,325.42) (15,447.24) (20,525.87) };

\addplot[
    color=purple,
    mark=diamond,
]
coordinates { (1,5.84) (3,9.88) (5,15.99) (10,27.73) (15,37.72) (20,47.74) };

\addplot[
    color=orange,
    mark=x,
]
coordinates { (1,0.16) (3,1.26) (5,2.73) (10,5.88) (15,6.36) (20,7.50) };
    
\end{axis}
\end{tikzpicture}
\begin{tikzpicture}[scale=0.3]
\begin{axis}[
     font= \LARGE,
    ylabel={Mitigation},
    xlabel={$k$},
    xmin=1, xmax=20,
    ymin=0, ymax=195,
    xtick={1,3,5,10,15,20},
    ymajorgrids=true,
    grid style=dashed,
]

\addplot[
    color=blue,
    mark=square,
]
coordinates { (1,13.69) (3,39.74) (5,60.62) (10,102.97) (15,136.17) (20,166.60) };

\addplot[
    color=cyan,
    mark=triangle,
]
coordinates { (1,21.23) (3,49.71) (5,66.91) (10,112.95) (15,161.05) (20,190.43) };
    
\addplot[
    color=red,
    mark=triangle,
]
coordinates { (1,13.14) (3,33.12) (5,47.03) (10,81.47) (15,105.28) (20,132.20) };

\addplot[
    color=purple,
    mark=diamond,
]
coordinates { (1,2.16) (3,6.88) (5,10.90) (10,20.90) (15,30.90) (20,40.92) };

\addplot[
    color=orange,
    mark=x,
]
coordinates { (1,0.18) (3,0.25) (5,0.77) (10,3.06) (15,2.18) (20,5.18) };
    
\end{axis}
\end{tikzpicture}
\begin{tikzpicture}[scale=0.3]
\begin{axis}[
     font= \LARGE,
    ylabel={$\beta$ lower bound},
    xlabel={$k$},
    xmin=1, xmax=20,
    ymin=0.59, ymax=0.71,
    xtick={1,3,5,10,15,20},
    legend pos=south east,
    ymajorgrids=true,
    grid style=dashed,
]

\addplot[
    color=teal,
    mark=square,
]
coordinates { (1,0.60) (3,0.60) (5,0.64) (10,0.65) (15,0.68) (20,0.68) };
\addlegendentry{FSTop}
    
\addplot[
    color=violet,
    mark=triangle,
]
coordinates { (1,0.64) (3,0.63) (5,0.64) (10,0.65) (15,0.69) (20,0.70) };
\addlegendentry{FSRan}
    
\end{axis}
\end{tikzpicture}
\begin{tikzpicture}[scale=0.3]
\begin{axis}[
     font= \LARGE,
    ylabel={$\beta$ lower bound},
    xlabel={$k$},
    xmin=1, xmax=20,
    ymin=0.62, ymax=0.84,
    xtick={1,3,5,10,15,20},
    ymajorgrids=true,
    grid style=dashed,
]

\addplot[
    color=teal,
    mark=square,
]
coordinates { (1,0.64) (3,0.68) (5,0.66) (10,0.67) (15,0.67) (20,0.67) };
    
\addplot[
    color=violet,
    mark=triangle,
]
coordinates { (1,0.76) (3,0.73) (5,0.78) (10,0.69) (15,0.70) (20,0.82) };
    
\end{axis}
\end{tikzpicture}
\caption{Mitigation on the Douban-Book (left) \& Douban-Movie (middle) datasets with $50$ ($200$) influential (random) fake seeds. Lower bound of $\beta$ for \emph{SA} (right).}
\label{fig:med_mit}
\end{figure*}



\begin{figure}
\centering
    \pgfplotsset{every axis legend/.append style={
        at={(0.7,1.07)},
        anchor=south}}
\begin{tikzpicture}[scale=0.3]
\begin{axis}[
    font=\LARGE,
    ylabel={Mitigation},
    xlabel={$k$},
    xmin=10, xmax=60,
    ymin=0, ymax=315,
    xtick={10,20,30,40,50,60},
    ymajorgrids=true,
    grid style=dashed,
    legend columns=5,
]

\addplot[
    color=blue,
    mark=square,
]
coordinates { (10,70.81) (20,131.31) (30,174.57) (40,226.38) (50,259.37) (60,309.85) };
\addlegendentry{\emph{NAMM}}

\addplot[
    color=cyan,
    mark=triangle,
]
coordinates { (10,32.99) (20,66.67) (30,95.24) (40,112.84) (50,141.98) (60,154.87) };
\addlegendentry{\emph{IMM}}
    
\addplot[
    color=red,
    mark=triangle,
]
coordinates { (10,25.88) (20,67.48) (30,97.24) (40,126.33) (50,152.66) (60,180.06) };
\addlegendentry{\textsc{Inf}}

\addplot[
    color=purple,
    mark=diamond,
]
coordinates { (10,20.00) (20,40.00) (30,60.00) (40,80.00) (50,100.00) (60,120.00) };
\addlegendentry{\textsc{Prox}}

\addplot[
    color=orange,
    mark=x,
]
coordinates { (10,0.59) (20,1.00) (30,0.82) (40,2.01) (50,1.34) (60,1.64) };
\addlegendentry{\textsc{Ran}}
    
\end{axis}
\end{tikzpicture}
\hspace{-8ex}
\begin{tikzpicture}[scale=0.3]
\begin{axis}[
    font=\LARGE,
    ylabel={Mitigation},
    xlabel={$k$},
    xmin=10, xmax=60,
    ymin=0, ymax=160,
    xtick={10,20,30,40,50,60},
    legend pos=north west,
    ymajorgrids=true,
    grid style=dashed,
]

\addplot[
    color=blue,
    mark=square,
]
coordinates { (10,35.71) (20,63.19) (30,86.73) (40,110.43) (50,130.58) (60,154.34) };

\addplot[
    color=cyan,
    mark=triangle,
]
coordinates { (10,5.53) (20,16.50) (30,22.40) (40,26.51) (50,28.33) (60,30.25) };
    
\addplot[
    color=red,
    mark=triangle,
]
coordinates { (10,6.41) (20,14.29) (30,21.89) (40,27.78) (50,28.58) (60,33.05) };

\addplot[
    color=purple,
    mark=diamond,
]
coordinates { (10,20.00) (20,40.00) (30,59.50) (40,73.74) (50,89.03) (60,104.16) };

\addplot[
    color=orange,
    mark=x,
]
coordinates { (10,0.02) (20,0.19) (30,0.43) (40,1.97) (50,0.40) (60,0.40) };
    
\end{axis}
\end{tikzpicture}
\hspace{-1ex}
\begin{tikzpicture}[scale=0.3]
\pgfplotsset{every axis legend/.append style={
        at={(0.4,1.07)},
        anchor=south}}
\begin{axis}[
    font=\LARGE,
    ylabel={$\beta$ lower bound},
    xlabel={$k$},
    xmin=10, xmax=60,
    ymin=0.60, ymax=0.80,
    xtick={10,20,30,40,50,60},
    ymajorgrids=true,
    grid style=dashed,
    legend columns=2,
]

\addplot[
    color=teal,
    mark=square,
]
coordinates { (10,0.62) (20,0.65) (30,0.63) (40,0.65) (50,0.63) (60,0.66) };
\addlegendentry{FSTop}
    
\addplot[
    color=violet,
    mark=triangle,
]
coordinates { (10,0.74) (20,0.75) (30,0.75) (40,0.76) (50,0.75) (60,0.77) };
\addlegendentry{FSRan}
    
\end{axis}
\end{tikzpicture}
    
\caption{Mitigation on the DBLP dataset with $50$ ($200$) influential (random) fake seeds and lower bound of $\beta$.}
\label{fig:dblp_mit_beta}
\end{figure}
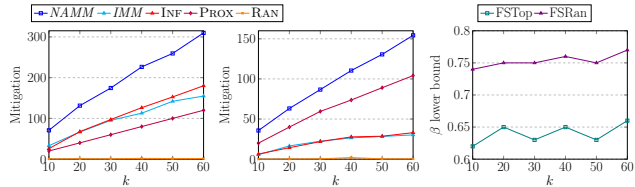




\begin{figure}
\centering
\hspace{-1ex}
\begin{tikzpicture}[scale=0.25]
    \begin{axis}[
    font=\LARGE,
    ybar stacked,
    xlabel={$k$},
    y tick label style={/pgf/number format/.cd,precision=4},
    xticklabels={1,3,5,10,15,20},
    xtick={1,...,6}
    ]
    \addplot coordinates
{(1,0.99595) (2,0.99753) (3,0.99775) (4,0.99882) (5,0.99929) (6,0.99957)};
    \addplot coordinates
{(1,0.00405) (2,0.00247) (3,0.00225) (4,0.00117) (5,0.00071) (6,0.00043)}; \end{axis}
\end{tikzpicture}
\hspace{-1ex}
\begin{tikzpicture}[scale=0.25]
    \begin{axis}[
    font=\LARGE,
    ybar stacked,
    xlabel={$k$},
    y tick label style={/pgf/number format/.cd,precision=4},
    xticklabels={1,3,5,10,15,20},
    xtick={1,...,6}
    ]
    \addplot coordinates
{(1,0.99955) (2,0.99911) (3,0.99885) (4,0.99820) (5,0.99867) (6,0.99901)};
    \addplot coordinates
{(1,0.00045) (2,0.00089) (3,0.00115) (4,0.00180) (5,0.00133) (6,0.00099)}; \end{axis}
\end{tikzpicture}
\hspace{-1ex}
\begin{tikzpicture}[scale=0.25]
    \begin{axis}[
    font=\LARGE,
    ybar stacked,
    xlabel={$k$},
    y tick label style={/pgf/number format/.cd,precision=4},
    xticklabels={1,3,5,10,15,20},
    xtick={1,...,6}
    ]
    \addplot coordinates
{(1,0.81808) (2,0.81553) (3,0.84325) (4,0.86696) (5,0.87415) (6,0.87939)};
    \addplot coordinates
{(1,0.18192) (2,0.18447) (3,0.15675) (4,0.13304) (5,0.12585) (6,0.12061)}; \end{axis}
\end{tikzpicture}
\hspace{-1ex}
\begin{tikzpicture}[scale=0.25]
    \begin{axis}[
    font=\LARGE,
    ybar stacked,
    xlabel={$k$},
    y tick label style={/pgf/number format/.cd,precision=4},
    xticklabels={1,3,5,10,15,20},
    xtick={1,...,6}
    ]
    \addplot coordinates
{(1,0.94065) (2,0.93110) (3,0.94025) (4,0.94632) (5,0.94853) (6,0.95260)};
    \addplot coordinates
{(1,0.05935) (2,0.06890) (3,0.05975) (4,0.05368) (5,0.05147) (6,0.04740)}; \end{axis}
\end{tikzpicture}
\caption{Reward breakdown on Gnutella (left) \& Douban-Book (right) with influential \& random fake seeds respectively.}
\label{fig:small_med_reward_breakdown}
\end{figure}
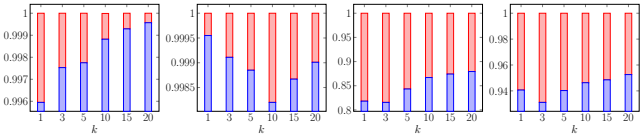


\stitle{Mitigation Results.} Across all datasets our \emph{NAMM} algorithm significantly outperforms the baselines in terms of mitigation achieved. \edit{Among the baselines, \emph{IMM} typically outperforms the rest and even produces the best mitigation on Douban-Movie under random fake seeds illustrating that, in some specific scenarios, a ``blanket'' approach that tries to spread the truth blindly to as much of the network as possible may outperform a more targeted approach.} Figures \ref{fig:small_mit}, \ref{fig:med_mit}, \& \ref{fig:dblp_mit_beta} also show lower bounds on the data-dependent approximation guarantee (Lemma \ref{lem:sandwich}) achieved by \emph{NAMM}. Specifically, we compute $\frac{\mu(S^U)}{\overline\mu(S^U)}$ which acts as a lower bound on the data-dependent guarantee of \emph{NAMM}. Across all datasets we observe that $\beta > 0.6$ and that $\beta$ is typically better when the fake seeds are chosen at random. \edit{In Figure \ref{fig:small_med_reward_breakdown} we plot the reward breakdown achieved by \emph{NAMM} on Gnutella and Douban-Book where the blue (red) component refers to the the fraction of nodes for which a reward of 2 (1) was achieved. We observe a strong tendency for reward 2 nodes to dominate the total reward breakdown. In particular, on Gnutella, over $99\%$ of the reward is due to winning outright under both influential and random fake seeds.}



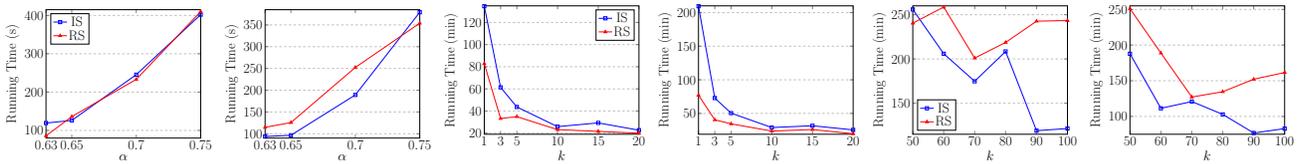
\begin{figure*}
\centering

\begin{tikzpicture}[scale=0.3]
\begin{axis}[
    font=\LARGE,
    ylabel={Running Time (s)},
    xlabel={$\alpha$},
    xmin=0.63, xmax=0.75,
    ymin=80, ymax=415,
    xtick={0.63,0.65,0.7,0.75},
    legend pos=north west,
    ymajorgrids=true,
    grid style=dashed,
]

\addplot[
    color=blue,
    mark=square,
]
coordinates { (0.63,119.24) (0.65,126.17) (0.7,245.45) (0.75,402.08) };
\addlegendentry{IS}
    
\addplot[
    color=red,
    mark=triangle,
]
coordinates { (0.63,85.96) (0.65,136.02) (0.7,233.22) (0.75,410.05) };
\addlegendentry{RS}
    
\end{axis}
\end{tikzpicture}
\begin{tikzpicture}[scale=0.3]
\begin{axis}[
    font=\LARGE,
    ylabel={Running Time (s)},
    xlabel={$\alpha$},
    xmin=0.63, xmax=0.75,
    ymin=90, ymax=385,
    xtick={0.63,0.65,0.7,0.75},
    ymajorgrids=true,
    grid style=dashed,
]

\addplot[
    color=blue,
    mark=square,
]
coordinates { (0.63,93.58) (0.65,96.44) (0.7,189.15) (0.75,379.16) };
    
\addplot[
    color=red,
    mark=triangle,
]
coordinates { (0.63,114.73) (0.65,125.62) (0.7,252.33) (0.75,354.04) };
    
\end{axis}
\end{tikzpicture}
\begin{tikzpicture}[scale=0.3]
\begin{axis}[
    font=\LARGE,
    ylabel={Running Time (min)},
    xlabel={$k$},
    xmin=1, xmax=20,
    ymin=19, ymax=135,
    xtick={1,3,5,10,15,20},
    legend pos=north east,
    ymajorgrids=true,
    grid style=dashed,
]

\addplot[
    color=blue,
    mark=square,
]
coordinates { (1,134.69) (3,61.39) (5,43.64) (10,25.88) (15,29.25) (20,22.66) };
\addlegendentry{IS}
    
\addplot[
    color=red,
    mark=triangle,
]
coordinates { (1,82.78) (3,33.22) (5,34.90) (10,23.19) (15,21.59) (20,19.93) };
\addlegendentry{RS}
    
\end{axis}
\end{tikzpicture}
\begin{tikzpicture}[scale=0.3]
\begin{axis}[
    font=\LARGE,
    ylabel={Running Time (min)},
    xlabel={$k$},
    xmin=1, xmax=20,
    ymin=19, ymax=210,
    xtick={1,3,5,10,15,20},
    ymajorgrids=true,
    grid style=dashed,
]

\addplot[
    color=blue,
    mark=square,
]
coordinates { (1,209.43) (3,72.89) (5,50.43) (10,28.98) (15,31.65) (20,25.48) };
    
\addplot[
    color=red,
    mark=triangle,
]
coordinates { (1,77.05) (3,40.48) (5,34.48) (10,23.73) (15,26.01) (20,19.76) };
    
\end{axis}
\end{tikzpicture}
\begin{tikzpicture}[scale=0.3]
\begin{axis}[
    font=\LARGE,
    ylabel={Running Time (min)},
    xlabel={$k$},
    xmin=50, xmax=100,
    ymin=115, ymax=260,
    xtick={50,60,70,80,90,100},
    legend pos=south west,
    ymajorgrids=true,
    grid style=dashed,
]

\addplot[
    color=blue,
    mark=square,
]
coordinates { (50,255.96) (60,205.93) (70,174.82) (80,208.51) (90,119.03) (100,121.45) };
\addlegendentry{IS}
    
\addplot[
    color=red,
    mark=triangle,
]
coordinates { (50,240.44) (60,258.68) (70,201.04) (80,218.70) (90,242.69) (100,243.40) };
\addlegendentry{RS}
    
\end{axis}
\end{tikzpicture}
\begin{tikzpicture}[scale=0.3]
\begin{axis}[
    font=\LARGE,
    ylabel={Running Time (min)},
    xlabel={$k$},
    xmin=50, xmax=100,
    ymin=75, ymax=255,
    xtick={50,60,70,80,90,100},
    ymajorgrids=true,
    grid style=dashed,
]

\addplot[
    color=blue,
    mark=square,
]
coordinates { (50,187.82) (60,111.39) (70,120.85) (80,102.84) (90,76.74) (100,82.90) };
    
\addplot[
    color=red,
    mark=triangle,
]
coordinates { (50,251.00) (60,188.82) (70,127.17) (80,134.61) (90,152.16) (100,161.53) };
    
\end{axis}
\end{tikzpicture}
    
\caption{Running time (left) under varying $\alpha$'s on Flixster with $k=20$, (middle) on Douban-Movie and (right) DBLP with $50$ ($200$) influential (random) fake seeds.}
\label{fig:runtime}
\end{figure*}


\stitle{Running Time.} Next, we investigate the running time achieved when leveraging importance sampling (IS) compared to rejection sampling (RS). We observe three interesting parameter domains exhibiting different behaviours. 

First, on the medium datasets (middle of Figure \ref{fig:runtime}), we observe that RS is typically faster than IS. To explain, consider the competing mechanisms at play that contribute to running time. On the one hand, RS typically requires more iterations of \emph{NAMM} to terminate due to the possibility of generating ``empty'' RDR sets that do not contribute reward signal. On the other hand, IS incurs additional runtime overhead as it is required to maintain and manage significantly more state than RS. Recall, under IS, the root of the RDR set is selected from $R_F^X$. Thus, in order to faithfully represent the possible world, the traversal tree from this first phase must be constructed and stored for every sample generated under IS. Further, in order to ensure efficient edge liveness lookups in the backward traversal, the edges traversed by $F$ must be sorted, incurring additional overheard. Note, we also tested storing these live edges in a hashmap, but observed an increase in runtime.

Interestingly, this trend is reversed in one parameter domain and mixed in another. First, while \emph{NAMM} terminates as soon as a $(1-1/e-\epsilon)$ approximation guarantee is achieved for each of the upper and lower bounding functions, the algorithm can be run beyond such a threshold. In the left plots of Figure \ref{fig:runtime} we observe that there is not a clear winner when the guarantee $\alpha$ is increased beyond $(1-1/e-\epsilon)$. Second, on the large dataset (right of Figure \ref{fig:runtime}) we see that IS outperforms RS for larger seed set sizes. In this regime, the cost of additional iterations of \emph{NAMM} outweighs the cost of additional state management.

\begin{table}
\caption{Effect of varying meeting length on Flixster.}
\centering
\resizebox{\linewidth}{!}{%
\begin{tabular}{|c||c|c|c|c||c|c|c|c|}
\hline
& \multicolumn{4}{c||}{FSTop} & \multicolumn{4}{c|}{FSRan} \\ \hline
ML & \emph{NAMM} & \textsc{Inf} & \textsc{Prox} & \textsc{Ran} & \emph{NAMM} & \textsc{Inf} & \textsc{Prox} & \textsc{Ran} \\ \hline
6 & (709.34) & (461.70) & (162.93) & (35.72) & (186.43) & (101.63) & (74.56) & (8.75) \\ \hline
5 & +7.81 & +14.19 & +2.56 & -1.40 & +2.12 & +2.95 & +0.22 & -1.51 \\ \hline
4 & +17.17 & +24.01 & +2.68 & -9.14 & +4.31 & +6.41 & -0.03 & +1.63 \\ \hline
3 & +14.11 & +30.53 & +2.44 & +9.18 & +2.5 & +6.95 & +0.24 & -0.08 \\ \hline
2 & +39.62 & +50.85 & +6.47 & +33.22 & +8.67 & +12.58 & +0.40 & -1.14 \\ \hline
1 & +76.38 & +117.34 & +13.39 & -29.27 & +17.63 & +27.29 & +0.54 & +4.59 \\ \hline
\end{tabular}}
\label{tbl:abl_ml}
\end{table}

\stitle{\edit{Alternative Parameter Settings.}} We conduct experiments that consider how the ability to mitigate the spread of fake news is impacted by varying the reading probabilities, activation window lengths and meeting lengths. In particular, we are interested in determining if mitigation can improve if users are more likely to click through and read an article, spend more time reading and considering the content and/or are exposed to true information with a similar propagation rate as false information.

Table \ref{tbl:abl_ml} shows the effects of reducing the meeting length disadvantage incurred by campaign $M$. In particular, we see that as the propagation rate of truth and misinformation approach an equal footing, there is a steady increase in the mitigation achieved by \emph{NAMM} and \textsc{Inf}. Furthermore, we observe that the marginal gain achieved increases as the propagation rate of campaign $M$ approaches that of $F$. \edit{Meanwhile, we observe that the \textsc{Random} and \textsc{Prox} baselines do not exhibit monotone behaviour when varying the meeting length.} Table \ref{tbl:abl_aw} shows the effects of elongating the AW length. We see a similar trend where the mitigation gain of \emph{NAMM} grows at an increasing rate with longer AW lengths. \edit{Meanwhile, the baselines all exhibit non-monotone behaviour with varying AW lengths and lower total gains.}



Finally, we did not observe a noticeable change in mitigation when varying the reading probability. To explain, notice that even when the reading probability succeeds (i.e., AW length is positive), the AW length sampled is unlikely to exceed a few hops in the propagation model. As such, there is little effect on the overall mitigation as the reading probability goes to 1.

\edit{In Figures \ref{fig:small_mit_ego} and \ref{fig:small_mit_fixed} we show the mitigation achieved under two alternative models: ego-centric meeting events and fixed edge probabilities respectively. First, under ego-centric meeting events, we consider $m(u,v) = c / (d_{out}(u) + c)$, since it is reasonable to deem that the more friends $u$ has, the less probable that $u$ could meet a certain individual in one time unit. Here $c$ is a smoothing constant and we set it according to \cite{chen2012time} which introduced meeting events in a single campaign setting. We find that the mitigation behaviour closely matches the results without ego-centric meeting events (Figure \ref{fig:small_mit}) both in terms of outperforming the baselines and in absolute reward values. Second, we consider fixing edge probabilities to $p = 0.1$ and continue to observe superior performance by \emph{NAMM}.}

\begin{table}
\caption{Effect of varying activation window length on Flixster.}
\centering
\resizebox{\linewidth}{!}{%
\begin{tabular}{|c||c|c|c|c||c|c|c|c|}
\hline
& \multicolumn{4}{c||}{FSTop} & \multicolumn{4}{c|}{FSRan} \\ \hline
AW & \emph{NAMM} & \textsc{Inf} & \textsc{Prox} & \textsc{Ran} & \emph{NAMM} & \textsc{Inf} & \textsc{Prox} & \textsc{Ran} \\ \hline
30 & (719.47) & (374.44) & (150.28) & (25.52) & (191.36) & (74.87) & (70.30) & (3.71) \\ \hline
60 & +10.51 & +12.72 & -0.76 & -10.04 & +0.12 & -0.03 & +0.45 & +4.87 \\ \hline
120 & +24.41 & -8.10 & -0.51 & +7.39 & +4.47 & +1.87 & -0.23 & -3.73 \\ \hline
240 & +48.56 & +7.94 & +2.23 & +4.68 & +8.6 & +1.95 & +0.53 & +0.47 \\ \hline
480 & +73.59 & +15.70 & +0.24 & -3.70 & +17.07 & +3.06 & -0.71 & +1.47 \\ \hline
\end{tabular}}
\label{tbl:abl_aw}
\end{table}


\begin{table}
\tiny
\caption{\edit{Jaccard similarity of seed sets for ground truth, 5/10\% perturbed, and no temporal parameter values on Flixster.}}
\centering
\begin{tabular}{|c||c|c|c||c|c|c|}
\hline
& \multicolumn{3}{c||}{FSTop} & \multicolumn{3}{c|}{FSRan} \\ \hline
k & P5 & P10 & CIC & P5 & P10 & CIC \\ \hline
1 & 1.00 & 1.00 & 1.00 & 1.00 & 1.00 & 1.00 \\ \hline
3 & 1.00 & 1.00 & 0.50 & 1.00 & 1.00 & 0.50 \\ \hline
5 & 1.00 & 1.00 & 0.55 & 1.00 & 1.00 & 0.67 \\ \hline
10 & 1.00 & 0.91 & 0.82 & 1.00 & 1.00 & 0.67 \\ \hline
15 & 0.88 & 0.94 & 0.67 & 0.94 & 0.94 & 0.67 \\ \hline
20 & 1.00 & 0.86 & 0.74 & 0.86 & 0.86 & 0.63 \\ \hline
\end{tabular}
\label{tbl:robust_flix}
\end{table}

\begin{table}
\tiny
\caption{\edit{Jaccard similarity of seed sets for ground truth, 5/10\% perturbed, and no temporal parameter values on Douban-Book.}}
\centering
\begin{tabular}{|c||c|c|c||c|c|c|}
\hline
& \multicolumn{3}{c||}{FSTop} & \multicolumn{3}{c|}{FSRan} \\ \hline
k & P5 & P10 & CIC & P5 & P10 & CIC \\ \hline
1 & 1.00 & 1.00 & 1.00 & 1.00 & 1.00 & 1.00 \\ \hline
3 & 1.00 & 1.00 & 1.00 & 1.00 & 1.00 & 0.50 \\ \hline
5 & 1.00 & 1.00 & 1.00 & 0.84 & 0.84 & 0.67 \\ \hline
10 & 0.91 & 0.91 & 0.67 & 1.00 & 1.00 & 0.67 \\ \hline
15 & 1.00 & 1.00 & 0.94 & 1.00 & 1.00 & 0.67 \\ \hline
20 & 0.90 & 0.90 & 0.74 & 0.95 & 0.90 & 0.78 \\ \hline
\end{tabular}
\label{tbl:robust_db}
\end{table}

\stitle{Robustness.} \edit{In real deployment, we may not have the exact ground truth values of the parameters. How robust are our solutions in the face of possibly imperfect temporal parameter settings? Also, how do our solutions compare with models such as CIC which have no temporal parameters? In Tables \ref{tbl:robust_flix} and \ref{tbl:robust_db} we show the Jaccard similarity between seed sets selected under a fixed ground truth, corresponding to the default parameter settings of \model, and those selected when the temporal parameters are perturbed and/or ignored. Specifically, we generate solutions after perturbing each of the meeting length, activation window, reading probability, and base propagation rate parameters with 5 \& 10\% Gaussian noise (P5 and P10). Additionally, we consider solutions generated with all temporal parameters ``turned off', i.e.,\ all meetings lengths set to $1$ and activation window lengths of $0$, which reduces \model to CIC. We observe that the perturbed solutions retain high similarity with the ground truth seed sets across both influential and random fake seeds. Furthermore, the similarity is always greater compared to the seed sets generated under CIC, highlighting the importance of capturing the differential propagation rates of truth and misinformation. Finally, the relative drop in mitigation for solutions generated under CIC is up to $20\%$ and $4.62\%$ on Flixster and Douban-Book respectively.}



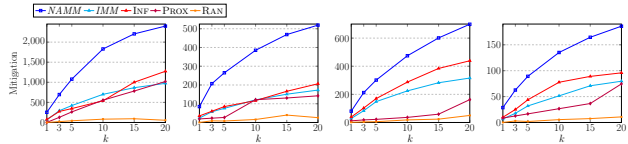
\begin{figure}
\centering
\begin{tikzpicture}[scale=0.23]
\pgfplotsset{every axis legend/.append style={
        at={(0.7,1.07)},
        anchor=south}}
\begin{axis}[
    font=\LARGE,
    ylabel={Mitigation},
    xlabel={$k$},
    xmin=1, xmax=20,
    ymin=0, ymax=2450,
    xtick={1,3,5,10,15,20},
    ymajorgrids=true,
    grid style=dashed,
    legend columns=5,
]

\addplot[
    color=blue,
    mark=square,
]
coordinates { (1,253.93) (3,695.08) (5,1078.85) (10,1830.11) (15,2205.09) (20,2399.63) };
\addlegendentry{\emph{NAMM}}

\addplot[
    color=cyan,
    mark=triangle,
]
coordinates { (1,63.87) (3,293.57) (5,421.65) (10,700.70) (15,866.01) (20,974.59) };
\addlegendentry{\emph{IMM}}
    
\addplot[
    color=red,
    mark=triangle,
]
coordinates { (1,75.48) (3,274.84) (5,346.51) (10,539.79) (15,1001.71) (20,1274.75) };
\addlegendentry{\textsc{Inf}}

\addplot[
    color=purple,
    mark=diamond,
]
coordinates { (1,2.00) (3,128.21) (5,263.60) (10,555.17) (15,782.68) (20,1020.21) };
\addlegendentry{\textsc{Prox}}

\addplot[
    color=orange,
    mark=x,
]
coordinates { (1,0.39) (3,22.07) (5,42.72) (10,82.39) (15,93.35) (20,55.62) };
\addlegendentry{\textsc{Ran}}
    
\end{axis}
\end{tikzpicture}
\hspace{-6ex}
\begin{tikzpicture}[scale=0.23]
\begin{axis}[
    font=\LARGE,
    xlabel={$k$},
    xmin=1, xmax=20,
    ymin=0, ymax=525,
    xtick={1,3,5,10,15,20},
    ymajorgrids=true,
    grid style=dashed,
]

\addplot[
    color=blue,
    mark=square,
]
coordinates { (1,84.82) (3,207.10) (5,265.17) (10,385.12) (15,468.98) (20,519.28) };

\addplot[
    color=cyan,
    mark=triangle,
]
coordinates { (1,23.69) (3,57.27) (5,75.57) (10,120.77) (15,150.88) (20,172.84) };
    
\addplot[
    color=red,
    mark=triangle,
]
coordinates { (1,35.00) (3,59.92) (5,84.75) (10,116.73) (15,166.39) (20,206.55) };

\addplot[
    color=purple,
    mark=diamond,
]
coordinates { (1,19.01) (3,22.93) (5,26.77) (10,121.53) (15,130.67) (20,142.22) };

\addplot[
    color=orange,
    mark=x,
]
coordinates { (1,0.25) (3,8.85) (5,8.01) (10,15.96) (15,39.53) (20,25.82) };
    
\end{axis}
\end{tikzpicture}
\begin{tikzpicture}[scale=0.23]
\begin{axis}[
    font=\LARGE,
    xlabel={$k$},
    xmin=1, xmax=20,
    ymin=0, ymax=700,
    xtick={1,3,5,10,15,20},
    ymajorgrids=true,
    grid style=dashed,
]

\addplot[
    color=blue,
    mark=square,
]
coordinates { (1,80.52) (3,212.59) (5,301.56) (10,474.49) (15,601.45) (20,698.76) };

\addplot[
    color=cyan,
    mark=triangle,
]
coordinates { (1,28.76) (3,83.53) (5,148.54) (10,224.78) (15,282.65) (20,316.01) };
    
\addplot[
    color=red,
    mark=triangle,
]
coordinates { (1,38.88) (3,102.62) (5,170.17) (10,287.37) (15,384.39) (20,438.96) };

\addplot[
    color=purple,
    mark=diamond,
]
coordinates { (1,11.70) (3,18.24) (5,22.27) (10,36.26) (15,59.07) (20,162.63) };

\addplot[
    color=orange,
    mark=x,
]
coordinates { (1,1.06) (3,3.91) (5,5.70) (10,18.41) (15,23.80) (20,50.09) };
    
\end{axis}
\end{tikzpicture}
\begin{tikzpicture}[scale=0.23]
\begin{axis}[
    font=\LARGE,
    xlabel={$k$},
    xmin=1, xmax=20,
    ymin=0, ymax=190,
    xtick={1,3,5,10,15,20},
    ymajorgrids=true,
    grid style=dashed,
]

\addplot[
    color=blue,
    mark=square,
]
coordinates { (1,28.71) (3,62.63) (5,89.09) (10,135.13) (15,164.50) (20,185.90) };

\addplot[
    color=cyan,
    mark=triangle,
]
coordinates { (1,6.67) (3,17.75) (5,31.94) (10,51.41) (15,70.85) (20,79.66) };
    
\addplot[
    color=red,
    mark=triangle,
]
coordinates { (1,10.07) (3,25.11) (5,43.82) (10,77.66) (15,88.97) (20,95.83) };

\addplot[
    color=purple,
    mark=diamond,
]
coordinates { (1,8.61) (3,12.62) (5,16.61) (10,26.62) (15,36.60) (20,74.49) };

\addplot[
    color=orange,
    mark=x,
]
coordinates { (1,0.07) (3,3.00) (5,1.91) (10,5.52) (15,7.49) (20,10.78) };
    
\end{axis}
\end{tikzpicture}
\caption{\edit{Mitigation under ego-centric meeting lengths on the Gnutella (left) \& Flixster (right) datasets.}}
\label{fig:small_mit_ego}
\end{figure}


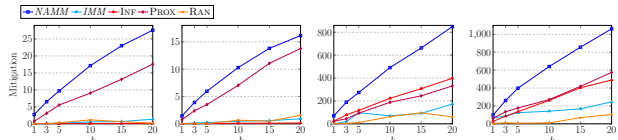
\begin{figure}
\centering
\begin{tikzpicture}[scale=0.23]
\pgfplotsset{every axis legend/.append style={
        at={(0.7,1.07)},
        anchor=south}}
\begin{axis}[
    font=\LARGE,
    ylabel={Mitigation},
    xlabel={$k$},
    xmin=1, xmax=20,
    ymin=0, ymax=29,
    xtick={1,3,5,10,15,20},
    ymajorgrids=true,
    grid style=dashed,
    legend columns=5,
]

\addplot[
    color=blue,
    mark=square,
]
coordinates { (1,2.73) (3,6.51) (5,9.70) (10,17.15) (15,23.00) (20,27.60) };
\addlegendentry{\emph{NAMM}}

\addplot[
    color=cyan,
    mark=triangle,
]
coordinates { (1,0.06) (3,0.14) (5,0.17) (10,0.58) (15,0.67) (20,1.36) };
\addlegendentry{\emph{IMM}}
    
\addplot[
    color=red,
    mark=triangle,
]
coordinates { (1,0.01) (3,0.01) (5,0.02) (10,0.10) (15,0.12) (20,0.14) };
\addlegendentry{\textsc{Inf}}

\addplot[
    color=purple,
    mark=diamond,
]
coordinates { (1,0.72) (3,3.13) (5,5.55) (10,9.06) (15,13.13) (20,17.59) };
\addlegendentry{\textsc{Prox}}

\addplot[
    color=orange,
    mark=x,
]
coordinates { (1,0.00) (3,0.00) (5,0.33) (10,1.17) (15,0.66) (20,0.18) };
\addlegendentry{\textsc{Ran}}
    
\end{axis}
\end{tikzpicture}
\hspace{-6ex}
\begin{tikzpicture}[scale=0.23]
\begin{axis}[
    font=\LARGE,
    xlabel={$k$},
    xmin=1, xmax=20,
    ymin=0, ymax=18,
    xtick={1,3,5,10,15,20},
    ymajorgrids=true,
    grid style=dashed,
]

\addplot[
    color=blue,
    mark=square,
]
coordinates { (1,1.43) (3,3.91) (5,5.97) (10,10.29) (15,13.80) (20,16.13) };

\addplot[
    color=cyan,
    mark=triangle,
]
coordinates { (1,0.04) (3,0.16) (5,0.22) (10,0.48) (15,0.58) (20,0.92) };
    
\addplot[
    color=red,
    mark=triangle,
]
coordinates { (1,0.00) (3,0.02) (5,0.04) (10,0.11) (15,0.12) (20,0.15) };

\addplot[
    color=purple,
    mark=diamond,
]
coordinates { (1,0.77) (3,2.39) (5,3.52) (10,7.04) (15,11.07) (20,13.77) };

\addplot[
    color=orange,
    mark=x,
]
coordinates { (1,0.00) (3,0.06) (5,0.02) (10,0.67) (15,0.48) (20,1.57) };
    
\end{axis}
\end{tikzpicture}
\begin{tikzpicture}[scale=0.23]
\begin{axis}[
    font=\LARGE,
    xlabel={$k$},
    xmin=1, xmax=20,
    ymin=0, ymax=860,
    xtick={1,3,5,10,15,20},
    ymajorgrids=true,
    grid style=dashed,
]

\addplot[
    color=blue,
    mark=square,
]
coordinates { (1,70.05) (3,188.85) (5,273.74) (10,491.22) (15,662.76) (20,849.49) };

\addplot[
    color=cyan,
    mark=triangle,
]
coordinates { (1,0.00) (3,19.86) (5,97.05) (10,69.37) (15,91.78) (20,172.28) };
    
\addplot[
    color=red,
    mark=triangle,
]
coordinates { (1,22.94) (3,78.09) (5,117.73) (10,222.98) (15,307.73) (20,398.03) };

\addplot[
    color=purple,
    mark=diamond,
]
coordinates { (1,23.43) (3,43.60) (5,93.18) (10,188.00) (15,244.52) (20,331.44) };

\addplot[
    color=orange,
    mark=x,
]
coordinates { (1,0.08) (3,3.84) (5,11.07) (10,62.92) (15,95.65) (20,60.66) };
    
\end{axis}
\end{tikzpicture}
\begin{tikzpicture}[scale=0.23]
\begin{axis}[
    font=\LARGE,
    xlabel={$k$},
    xmin=1, xmax=20,
    ymin=0, ymax=1100,
    xtick={1,3,5,10,15,20},
    ymajorgrids=true,
    grid style=dashed,
]

\addplot[
    color=blue,
    mark=square,
]
coordinates { (1,96.62) (3,260.30) (5,398.58) (10,640.54) (15,856.62) (20,1061.06) };

\addplot[
    color=cyan,
    mark=triangle,
]
coordinates { (1,65.85) (3,86.55) (5,124.94) (10,141.20) (15,167.72) (20,244.90) };
    
\addplot[
    color=red,
    mark=triangle,
]
coordinates { (1,26.96) (3,85.97) (5,136.13) (10,262.27) (15,402.87) (20,490.15) };

\addplot[
    color=purple,
    mark=diamond,
]
coordinates { (1,61.91) (3,135.11) (5,176.95) (10,271.16) (15,418.52) (20,574.86) };

\addplot[
    color=orange,
    mark=x,
]
coordinates { (1,1.83) (3,11.83) (5,6.81) (10,10.55) (15,69.00) (20,105.01) };
    
\end{axis}
\end{tikzpicture}
\caption{\edit{Mitigation under fixed propagation probabilities on the Gnutella (left) \& Flixster (right) datasets.}}
\label{fig:small_mit_fixed}
\end{figure}
\section{Related Work}
\label{sec:related_work}

\stitle{Influence Maximization.} The IM problem was  formulated as a discrete optimization problem by Kempe et al.\ \cite{kempe2003} where the independent cascade (IC) and linear threshold (LT) models were introduced. Since then, various aspects of IM, and its variants, have been extensively studied (see \cite{Chen2013, li2018influence} for surveys\eat{ that cover this area}). State-of-the-art IM solutions \eat{One line of work} \cite{Tang2014, tang2015influence, nguyen2016stop, huang2017revisiting, tang2018online} rely on reverse sampling for  their efficiency. \eat{aimed to improve efficiency. and state-of-the-art solutions are based on reverse sampling. Next, we highlight the most relevant IM variants to our work.} Our edge-level time-delayed propagation is similar to the diffusion dynamics captured by the IC-M model of \cite{chen2012time} set in a single campaign propagation setting \eat{where the concept is used}  to model the \eat{effects of} log-in and log-out behaviour of users. The concept of a time-sensitive reward function was considered in \eat{variants of the IM problem}  \cite{liu2012time, khan2016towards} to better model time-sensitive information such as product sales. IM under competition is studied in \cite{Bharathi2007, Lin2015, lu2015competition} among others. Finally, models that distinguish between awareness and adoption have been considered in \cite{bhagat2012maximizing, lu2015competition}.

\stitle{Misinformation Mitigation.} The MM problem was first studied under an independent cascade model by Budak et al. \cite{budak2011limiting} and under a linear threshold model in \cite{he2012influence, Fan2013}. In both settings, the objective is shown to be monotone and submodular, thus the greedy algorithm provides a $(1 - 1 / e)$ approximation. Subsequently, there have been a number of works that either study variants of the classical MM problem or improve the running time of the greedy approach. The related problem of determining the budget required to reach a threshold mitigation level is investigated in \cite{Fan2013, pham2018targeted, pham2019minimum}.

Each of \cite{tong2018misinformation, tong2019beyond, simpson2020reverse, tong2017efficient, song2017temporal, saxena2020mitigating} considers MM variants under models that ignore the temporal nature of misinformation propagation. Roughly speaking, they parallel the improvements made in reverse sampling frameworks developed for the IM problem. 
Most relevant, \cite{song2017temporal} uses a competitive IC model augmented with meeting events that are shared by \emph{both} campaigns. Thus, the observed difference in propagation rates between fake and true information is not captured. Notably, none of the above studies incorporates activation windows or a reward function that is time-critical. 

Finally, related versions of the MM problem have been investigated by other communities including crowd-sourced mitigation \cite{kim2018leveraging, vo2018rise}, epidemiology \cite{simpson2016clearing, vu2017minimizing, khalil2014scalable, prakash2013fractional, zhang2015data, wu2021containment}, and ML \cite{tong2020stratlearner, farajtabar2017fake}. 
In particular, the epidemiology community focuses on \textit{preemptive} strategies, without propagation for the mitigating side, that limit the susceptibility of a network, \LL{and} consider propagation models without the temporal notions we study. The ML approaches aim to learn mitigation strategies trained on a fixed input graph which limits their transferability to new graph instances. Recently, Juul et al.\ \cite{juule2021comparing} repeat the comparison of the spread of fake and true news conducted by \cite{vosoughi2018spread} on two subsampled datasets with exactly the same size distribution. Interestingly, they find that {\sl under these conditions}, the propagation characteristics become indistinguishable, which has important consequences for misinformation detection. We note that our work focuses on the general case where the propagation differences observed by Vosoughi et al.\ \cite{vosoughi2018spread} do hold. In particular, the authors of \cite{juule2021comparing} confirm that false and true information propagate at different rates when the naturally occurring cascade size distributions remain unaltered. 


\section{Conclusion and Future Work}
\label{sec:conclusion}

In this paper, we address a major shortcoming of existing MM propagation models by introducing the \model model, which captures important temporal aspects of fake news diffusion and formulate a time-sensitive variant of the MM problem. We prove our mitigation objective is non-submodular and develop submodular upper and lower bounding functions to sandwich the objective and provide data-dependent approximation guarantees. Finally, we propose a reverse sampling framework that provides $(1 - 1/e - \epsilon)$-approximate solutions to our bounding functions and present an anytime version of our approach. \edit{Using experiments over five datasets, we demonstrate that our \emph{NAMM} algorithm outperforms various baselines including those that are oblivious to time-critical aspects.} 

\edit{The new techniques developed here are applicable to other problem settings. E.g., in the filter bubble problem, multiple conflicting opinions propagate in a social network and the goal is to ensure balanced exposure. Here, the set of users adopting any particular opinion is stochastic and thus our \emph{NAMM} algorithm can be applied to help obtain a more balanced exposure for the users. Further, revisit the ``classic'' competitive IM problem, where realistically the diffusion rates of companies/brands may differ owing to differences in reputation and marketing strategies. In this case, \emph{NAMM} can be applied for the competitive IM problem from a follower perspective. Notice that previous techniques for competitive IM do not apply to this setting.} 

\edit{For simplicity of exposition, we considered a fixed set of fake seeds. However, our proposed solution retains its guarantees even when the fake seeds are not known exactly but are chosen from a distribution. It is interesting to study the scenario where the fake seed set is dynamically growing. This is closely tied with adaptive influence maximization, which we leave for future work.}


\balance
\bibliographystyle{abbrv}
\bibliography{refs}

\end{document}